\documentclass[12pt,letterpaper,reqno,oneside]{article}
\usepackage{booktabs,amsmath,amsfonts,amssymb,graphicx,bm}
\usepackage[authoryear]{natbib}
\usepackage{relsize}
\usepackage{times}
\usepackage{dsfont}
\usepackage{etoolbox}
\usepackage{float}
\usepackage{epigraph}
\usepackage{ushort}
\usepackage{hyperref}
\usepackage{xcolor}
\usepackage{tikz}
\usepackage{relsize}
\usepackage{pgfplots}
\usepackage{setspace}
\usepackage{mathtools} 
\usepackage{bbm}
\usepackage{caption}
\usepackage{subcaption}
\usepackage{amsthm}
\usepackage{comment}
\usepackage{enumerate}
\usetikzlibrary{positioning}

\newtheorem{assumption}{Assumption}[]

\pgfplotsset{compat=1.10}
\usepgfplotslibrary{fillbetween}
\usepackage[margin=1.1in]{geometry}

\makeatletter
\renewcommand{\paragraph}{%
	\@startsection{paragraph}{4}%
	{\z@}{1.75ex \@plus 1ex \@minus .2ex}{-0.7em}%
	{\normalfont\normalsize\bfseries}%
}
\makeatother
\hypersetup{
  colorlinks   = true, 
  urlcolor     = blue, 
  linkcolor    = blue, 
  citecolor   = blue 
}

\newtheorem{theorem}{Theorem}
\newtheorem{claim}{Claim}
\newtheorem{corollary}{Corollary}
\newtheorem{lemma}{Lemma}
\newtheorem{proposition}{Proposition}
\newtheorem{obs}{Observation}

\newtheorem{definition}{Definition}
\newtheorem{numres}{Numerical Result}

\usepackage{datetime}
\newdateformat{datestyle}{\monthname[\THEMONTH], \THEYEAR}

\title{\scshape Disclosure and Incentives in Teams\footnote{Paula Onuchic: paula.onuchic@economics.ox.ac.uk. Jo\~{a}o Ramos: joao.ramos@usc.edu. We are grateful for detailed comments and suggestions from Nageeb Ali, Rohan Dutta, Debraj Ray, Ludvig Sinander, and Mark Whitmeyer, as well as the audience at the $2^{nd}$ Southeast Theory Festival at Nuffield College.}}

\author{%
\begin{tabular}{cc}
	Paula Onuchic
	& Jo\~{a}o Ramos \\
	University of Oxford
	& USC Marshall
\end{tabular}%
}
\date{\datestyle\today}

\begin{document}
\maketitle

\begin{abstract}
We consider a team-production environment where all participants are motivated by career concerns, and where a team's joint productive outcome may have different reputational implications for different team members. In this context, we characterize equilibrium disclosure of team-outcomes when team-disclosure choices aggregate individual decisions through some \emph{deliberation protocol}. In contrast with individual disclosure problems, we show that equilibria often involve partial disclosure. Furthermore, we study the effort-incentive properties of equilibrium disclosure strategies implied by different deliberation protocols; and show that the partial disclosure of team outcomes may improve individuals' incentives to contribute to the team. Finally, we study the design of deliberation protocols, and characterize productive environments where effort-incentives are maximized by unilateral decision protocols or more consensual deliberation procedures.
\end{abstract}

\section{Introduction}
Productive activities are increasingly conducted in teams. Startups are often founded by entrepreneurial partners,\footnote{For instance, $80\%$ of all billion-dollar companies launched since 2005 had two or more founders. See, for example, this report: https://alitamaseb.medium.com/land-of-the-super-founders-a-data-driven-approach-to-uncover-the-secrets-of-billion-dollar-a69ebe3f0f45} and even in large established firms, new products are mostly developed and proposed by teams built and empowered within the company.\footnote{Lazear and Shaw (2007) document that close to 80\% of US firms rely on self-managing teams in some capacity. Recent literature  also documents the rise of teamwork in scientific research. See for example Fortunato et al. 2018, Schwert 2021, and Jones (2021).} 
Teams vary in terms of how ownership and decision rights are allocated among team-members --- for example, some teams have decisions made by pre-assigned team-leaders while other teams' decisions may require a majority or even consensus. Naturally, the allocation of decision rights establishes the power-dynamics that determines the team's priorities and goals, thereby impacting team-members' willingness to collaborate. In this paper, we enrich the scope of authority to encompass team-members' rights over how to communicate the teams' productive outcomes.  We show that the allocation of such decision rights impacts how the team's outcomes are perceived by outside observers. Specifically, it determines how \emph{credit} and \emph{blame} for the team's successes and failures are distributed across team-members; and, through this channel, affects individuals' incentives to contribute to the team.

We consider a team-production environment where all participants are motivated by career concerns, and where a team's joint productive outcome may have different reputational implications for different team members. For example, a development team may be composed of engineers and marketers, so that a technically impressive, but ``badly packaged'' new product would be seen as a engineering success, but have negative reputational implications to the marketing team-members. In such an environment, we study the implications of varying the allocation of decision rights regarding the \emph{disclosure} of the team's outcomes to an outside observer. Our analysis establishes two types of results. First, we show that when disclosure decisions are made by teams --- using some \emph{deliberation protocol} that aggregates team members' individual decisions --- equilibria often involve partial disclosure. This is in direct contrast with the uniqueness of full-disclosure equilibrium in environments where disclosure is an individual decisions, as in the literature started by Grossman (1981) and Milgrom (1981). We show that such partial-disclosure equilibria are induced by the observer's inability to fully attribute blame for ``no disclosure'' to a specific team-member. This inability disrupts the \emph{unravelling logic} behind the uniqueness of the full-disclosure equilibrium in a single-agent environment.

Having characterized team-disclosure equilibria, we then establish results of a second kind, exploring the incentive properties of equilibrium disclosure strategies implied by different deliberation protocols. Specifically, we compare a deliberation protocol where team-members can unilaterally choose disclosure --- which in equilibrium induces full-disclosure --- to protocols in which disclosure requires a higher degree of consensus, which in equilibrium induce partial-disclosure. We show that more consensual protocols may provide incentives for team members to exert costly effort to improve the team's outcomes, above and beyond their incentives when outcomes are always revealed to the observer. Intuitively, when disclosure choices are made by a team rather than by an individual, effort exerted by team-member $A$ has two potential positive effects: it directly improves team-member $A$'s own outcomes, an effect also present when all team outcomes are disclosed to the observer; but additionally, by improving the outcomes of $A$'s partners, $A$'s effort decreases the probability that the team vetoes the disclosure of outcomes that are good news about $A$. 
\vskip10pt

\noindent \textbf{Disclosure in Teams.} In section \ref{sec:disc}, we introduce our team-disclosure environment.
A team is made up of $N\geqslant 2$ \emph{team-members}, who produce a team-outcome, drawn from a distribution known by all team-members and by an \emph{outside observer}.\footnote{At a later section, the team's outcome distribution is endogenized, and we focus on the effort incentives provided by different deliberation  protocols.} A team-outcome is a piece of (perhaps imperfect) \emph{hard evidence}, which implies some payoff for each team-member. If a realized outcome is seen by the outside observer, then they interpret it and form a belief about each team member's ability (or other reputational aspect of interest), which generates a payoff for each individual. 
After observing the outcome, the team --- through a deliberation protocol detailed below --- makes the binary decision to \emph{disclose} or \emph{not disclose} the team's outcome to the outside observer. If the outcome is disclosed, then each team member receives a payoff equal to their expected ability implied by the evidence. If instead the outcome is not disclosed, then the observer forms a rational belief about each team-member's ability, accounting for the circumstances that may have led to non-disclosure. Each team-member then receives a payoff equal to the observer's belief about their own expected ability.

The team's disclosure decision is made by aggregating all team-members' individual disclosure decisions through some \emph{deliberation process}. A deliberation process is an aggregation procedure that satisfies the following two basic properties: (i) it respects consensus, so that if all team members individually favor disclosing/not-disclosing the outcome, then that decision is implemented; and (ii) it is monotonic, so that if (non-) disclosure is implemented when a subset of the team-members each favor (non-) disclosure, then it is also implemented when a larger subset of team-members favors that decision. For example, a deliberation process may be such that one of the team-members is the leader, so that the team's disclosure decision is given by the leader's individual disclosure decision. Another possibility is that the team chooses disclosure if and only if at least $K\leqslant N$ team-members favor disclosure. These \emph{anonymous} deliberation protocols incorporate \emph{unilateral disclosure} if $K=1$, so that every team-member can unilaterally choose disclosure; and \emph{consensual disclosure} if $K=N$, so that the team's outcome is disclosed only if all team-members favor that decision.

The deliberation procedure affects the team's disclosure decision directly by establishing whose individual decisions are heeded by the team. Additionally, in equilibrium it \emph{indirectly} affects the team's decision, because it impacts the formation of the observer's beliefs upon seeing ``no disclosure.'' Suppose for example that a team has $N=2$ members and either team-member can unilaterally decide to disclose the team's outcome. Then, upon seeing that the team chose not to disclose their outcome, the observer understands that the realized outcome was ``bad news'' to both team-members (for otherwise one of them would have chosen to disclose it), and forms an unfavorable assessment of both team-members' abilities. Such skepticism from the observer then activates the usual ``unravelling logic'' and leads to the unique equilibrium being such that the team's outcomes are always disclosed. Now suppose instead that the deliberation protocol is such that both team-members can veto disclosure --- in other words, disclosure must be chosen consensually. In that case, upon seeing no disclosure, the observer cannot tell whether the outcome was ``bad news'' about team-member $1$, who therefore vetoed its disclosure, or about team-member $2$. Consequently, the observer's ``no-disclosure belief'' is not so skeptical about either team-member's ability; and therefore supports equilibria with partial-disclosure in which disclosed outcomes must have good-enough reputational implications to both team-members. 

Following this logic, Theorem \ref{th:1} and Proposition \ref{pr:1a} characterize the team-disclosure equilibrium set implied by different deliberation protocols.
We show that full-disclosure is the unique equilibrium outcome if and only if the disclosure protocol allows all team members to unilaterally disclose the team's outcome.\footnote{In section \ref{sec:eqdisc}, we also distinguish \emph{partial-disclosure} equilibria, in which at least one team-member has some of their good outcomes not disclosed, from \emph{interior} equilibria, in which \emph{all} team-members have some of their good outcomes not disclosed. Theorem \ref{th:1} also states that if the deliberation protocol is such that no team-member can unilaterally disclose the outcome, then all partial-disclosure equilibria are interior}  When a partial-disclosure equilibrium exists, it is such that disclosure happens when sufficiently many team members draw sufficiently good outcomes (outcomes that are better than the observer's belief implied by no disclosure). 
Even under deliberation protocols such that partial-disclosure equilibria exist, the equilibrium set always includes a full disclosure equilibrium, which is supported by very skeptical no-disclosure beliefs held by the observer off the equilibrium path. 

Theorem \ref{th:2} complements the characterization in Theorem \ref{th:1} by describing a refinement of the equilibrium set. To that end, we introduce a refinement criterion, bespoke to our team environment, to evaluate whether full disclosure is an equilibrium that is \emph{consistent with deliberation}.\footnote{The criterion is in the spirit of \emph{sequential equilibria}, as in Kreps and Wilson (1982), and other refinements proposed for communication games. Unlike these standard refinement criteria, our refinement is bespoke to an environment in which decisions are made by a team, and evaluates the plausibility of equilibrium outcomes with respect to the given deliberation protocol.} Intuitively, this criterion imposes the requirement that even off-path beliefs held by the outside observer should be justified by the aggregation of individual behavior of the team-members through the given deliberation protocol. 
Theorem \ref{th:2} establishes that full-disclosure is consistent with deliberation if and only if the deliberation protocol is such that disclosing requires \emph{no more consensus} than concealing. This property of the deliberation protocol is defined formally in the text. However, as an example, suppose that the deliberation protocol is anonymous, so that disclosure is chosen if at least $K\leqslant N$ team-members favor disclosure. Then we say that disclosure requires no more consensus than concealing if $K\leqslant N/2$.

\vskip10pt

\noindent\textbf{Deliberation and Incentives.} Having characterized equilibria of the team-disclosure game, in section \ref{sec:inc} we augment the environment with an initial stage in which team-members choose whether to exert costly effort to improve the team's outcome distribution. Importantly, each team-member's effort positively affects not only the value of the team-outcome for themselves (what we call the individual's \emph{own outcome}), but also the value of the team-outcome to the other team-members. At the initial stage, each team member individually chooses whether or not to exert effort, at a given cost. Each team-member's effort choice is not observed by the other participants, or by the outside observer. Instead, as before, all members of the team observe the realization of the team outcome --- drawn from a distribution which now depends on the team-members' effort profile --- and then the team chooses to disclose or not disclose the outcome to the observer, according to the deliberation protocol. Our main results in this section evaluate the effort-incentives provided by different deliberation protocols, through their effect on equilibrium disclosure of team-outcomes. 

Theorem \ref{th:3} characterizes each team-member's incentives to exert costly effort as a function of the foreseen disclosure strategy used by the team to communicate team-outcomes to the outside observer. Specifically, it decomposes the benefits of effort in two parts: first, an individual's effort positively affects the distribution of their own outcomes --- this effect is exactly the benefit from effort when the team uses a full-disclosure strategy. Additionally, if the team's disclosure strategy involves the concealment of some team-outcomes, then, by affecting the distribution of other team-members' outcomes, an individual's effort also impacts the correlation between disclosure and the individual's own outcome realization. Consequently, if the impact of a team-member's effort on the team-outcome distribution is such that it improves the correlation between their own outcome and the team's disclosure decision, then the strategic disclosure provides effort-incentives beyond those provided by full-disclosure.

 Using the characterization in Theorem \ref{th:3}, Propositions \ref{pr:3}-\ref{pr:5} rank the incentive power of different deliberation procedures. To that end, we say that a deliberation procedure $D$ dominates a deliberation procedure $D'$ if \emph{full-effort} --- that is, effort is exerted by all team-members --- can be implemented in equilibrium under procedure $D$ whenever it can be implemented in equilibrium under procedure $D'$. Moreover, we say that effort is \emph{self-improving} when each individual's effort improves the distribution of their own outcomes, but leaves the outcome distribution of other team-members unchanged. In contrast, effort is \emph{team-improving} if an individual's effort leaves their own outcome distribution unchanged, but improves the outcome distribution of other team-members. 
 
 Proposition \ref{pr:3} shows that, if effort is self-improving, then a deliberation protocol in which all team-members can unilaterally choose to disclose the team's outcome dominates all other deliberation protocols. Conversely, it states that if effort is team-improving, then a protocol where disclosure must be chosen by consensus, by supporting equilibria with partial-disclosure, dominates the unilateral disclosure protocol under which full-disclosure is the unique equilibrium of the team-disclosure stage. Intuitively, under the consensual disclosure deliberation protocol, each team-member has incentives to improve the outcomes of their partners, so as to avoid situations where the disclosure of their own good outcome realizations is vetoed by others. This feature of partial-disclosure disclosure equilibria creates strategic complementarities between team-members, thereby facilitating the implementation of full-effort.

Relying on the same intuition delineated in Theorem \ref{th:3}, Proposition \ref{pr:4} considers a broader class of ``effort types,'' and ranks the effort incentives provided by deliberation protocols that induce partial-disclosure equilibria against those provided by the unilateral disclosure protocol. Specifically, it shows that if each team-member's effort sufficiently improves the correlation between all team-member's outcomes, then all deliberation protocols dominate unilateral disclosure. Finally, Proposition \ref{pr:5} characterizes effort environments where an \emph{effective team leader} exists, in the sense that a deliberation protocol in which a particular team-member is given disclosure power dominates the unilateral disclosure protocol.

As a final exercise, in section \ref{sec:incbin} we study an example where all team-members' outcomes are binary. That is, for each team-member, every possible team-outcome is either interpreted as ``good news'' or as ``bad news.'' And moreover, we impose that the deliberation protocol is anonymous, so that the team-outcome is disclosed if and only if at least $K\leqslant N$ team-members favor its disclosure. In this context, Proposition \ref{pr:eff_bin} ranks the effort incentives provided by deliberation protocols in which disclosure requires different consensus levels $K$. We also study ``optimal deliberation'' in this binary environment, and show that the optimal required consensus level may be interior. The optimal required consensus increases as the effect of effort on the correlation between team-members' outcomes increases and as the effect of an individual's effort on the distribution of their own outcomes decrease. And the optimal consensus is a non-monotone function of the effect of an individual's effort on the distribution of other team-members' outcomes.


\subsection{Related Literature} 
The team-disclosure problem studied in the first part of the paper relates to a large literature on communication, and specifically to the literature on evidence disclosure and that on multi-sender communication. Within the literature on evidence disclosure, stemming mainly from Grossman (1981) and Milgrom (1981), we follow the Dye (1985) tradition in assuming that a piece of evidence --- in our case the team outcome --- can be either fully revealed to the outside observer or fully concealed. This evidence structure is in contrast with other communication models such as cheap talk as in Crawford and Sobel (1982) or Bayesian persuasion as in Kamenica and Gentzkow (2011). We contribute to this literature by studying a problem of \emph{team}-disclosure, where disclosure decisions are made by a group of agents through some deliberation procedure, rather than by a single individual. As mentioned, our characterization of team-disclosure equilibria contrasts with typical results in parallel individual-disclosure models, as it features partial evidence disclosure. We also connect the teams' strategy in our partial-disclosure equilibria to \emph{sanitization} strategies as in Dye (1985). 

Within the evidence disclosure literature, our environment most resembles models with multidimensional evidence --- such as Dziuda (2011) and Martini (2018) --- and multi-sender disclosure --- such as Baumann and Dutta (2022), and Hu and Sobel (2019).\footnote{There is also a large literature on multi-sender cheap talk and multi-sender information design --- see for example Battaglini (2002) and Gentzkow and Kamenica (2016) --- which is less closely related to our work.} In contrast with the former set of papers, we study disclosure by multiple senders, rather than individuals.\footnote{In section \ref{sec:eqd}, we briefly connect the characterization of team-disclosure equilibria in our setup to that of individual disclosure equilibria with multidimensional evidence as in Martini (2018).} And our model differs from the latter papers on disclosure by multiple senders in that we study disclosure decisions made by \emph{teams} through some cooperative deliberation process, rather than unilateral disclosures made by competing senders. 

Our paper also connects to a small literature relating disclosure and incentives. Ben-Porath, Dekkel, and Lipman (2018) show that in a Dye (1985) individual-disclosure environment, partial disclosure equilibria may incentivize the individual to favor risky projects, even at the expense of the project's overall expected value. Matthews and Postlewaite (1985), and more recently Shishkin (2021), Onuchic (2022), and Whitmeyer and Zhang (2022), study the effect of the evidence-disclosure equilibrium on an individual's incentives to acquire evidence. 

A closely related literature --- for example, Austen-Smith and Feddersen (2005), Gerardi and Yariv (2007, 2008), Levy (2007), Visser and Swank (2007), and more recently Bardhi and Bobkova (2023) --- study information acquisition and information aggregation in deliberative committees, under various voting and communication protocols as well as committee compositions. Our paper first departs from that literature in that we study a model of \emph{evidence disclosure} by a team, rather than an environment where an action choice is delegated to a committee.\footnote{Bardhi and Bobkova (2023) also study an environment where the committee discloses evidence to a principal. However, in their model, all acquired evidence is necessarily disclosed to the outside observer. Contrastingly, we study a problem where disclosure decisions are made through some deliberation protocol which aggregates the teams' interests.} More importantly, our paper differs from that literature in its focus:  while the  deliberative committees literature studies how different protocols fare in terms of information acquisition and aggregation, we characterize disclosure equilibria under various protocols and evaluate their power to incentivize team-members to put effort into a team project.

As mentioned, our characterization of partial-disclosure equilibria in the team context relies on the observer's inability to perfectly attribute credit for a team's successes and failures across the team-members. In that sense, our paper relates to a small literature that considers credit attribution in teams.  Onuchic \textcircled{r} Ray (2023),  Ray \textcircled{r} Robson (2018), and Ozerturk and Yildirim (2021) study team production with unequal credit attribution to team-members. In the latter two papers, the attribution of credit is endogenously based on estimates of individual contributions, which inefficiently affects individual effort decisions. But there are no reputational concerns, and credit attributed to each agent only determines their share in the physical outcome of the project. 

Finally, our paper also contributes to a nascent literature on games played by teams of players and mechanism design in teams. See, for example, Kim, Palfrey and Zeidel (2022), Hara (2022), Kuvalekar, Haghpanah, and Lipnowski (2022), and Battaglini and Palfrey (2023). An alternative way to interpret the deliberation protocol is as a set of rules to which the team pre-commits, specifying which team-members' ``incentive constraints'' must be respected by the team. We view our exercise in designing the optimal deliberation protocol as a first step towards further work on designing team-decision protocols.


\section{Team Disclosure}
\label{sec:disc}
\subsection{Environment}
A group of $N\geqslant 2$ agents makes up a team, whose outcome $\omega=(\omega_1,...,\omega_N)$ may be seen by an outside observer. The team's outcome is drawn from a distribution $F$ over a finite outcome space $\Omega\subset\mathbb{R}^N$.\footnote{In Appendix \ref{app:B1}, we discuss a variation of the benchmark model, in which outcomes are continuous.} For each $\omega\in\Omega$ and  $i\in N$, $\omega_i$ should be interpreted as the outcome of interest to team member $i$. When the outcome $\omega$ realizes, the team has some piece of \emph{hard evidence} that conveys precisely the outcome realization.\footnote{The evidence structure is as in Dye (1985), in that the evidence perfectly conveys the outcome realization to the observer. Unlike in Dye (1985), we assume that the team \emph{always} has access to hard evidence of the outcome.} Each team member $i$ then makes an \emph{individual disclosure decision} $x_i(\omega_i)\in[0,1]$, where $x_i(\omega_i)=1$ indicates that agent $i$ is in favor of the evidence being disclosed to the outside observer and $x_i(\omega_i)=0$ indicates they are against it; $x_i(\omega_i)\in(0,1)$ is the probability that agent $i$ favors the outcome's disclosure.\footnote{We require each team-member's individual disclosure strategy to depend only on their own outcome realization. We comment on that assumption in section \ref{sec:interp} and in Appendix \ref{app:B2}.}

The \emph{teams' disclosure decision} is then made according to some \emph{deliberation process}. A deliberation process $D:[0,1]^N\rightarrow[0,1]$ is procedure that aggregates individual disclosure decisions into a team disclosure decision. For example, $D(x)=1$ if and only if $x=(1,...,1)$ indicates that disclosure decisions are made consensually; and $D(x)=x_i$ for all $x$ indicates that agent $i$ is a team leader. The team's disclosure decision $d(\omega)=D(x_1(\omega_1),...,x_N(\omega_N))\in[0,1]$ represents the probability that the outcome $\omega$ is disclosed to the outside observer. We make the following assumptions about the deliberation process. 

\begin{assumption}
\label{as:1}
The deliberation process $D:[0,1]^N\rightarrow[0,1]$
\begin{enumerate}
\item Agrees with consensual team decisions: $D(x,...,x)=x$ for $x\in\{0,1\}$.
\item Is monotone: $x\geqslant x'$ implies $D(x)\geqslant D(x')$.
\item Is deterministic: 
$$D(x)\in\{0,1\}\text{ if }x_i\in\{0,1\}\text{ for every }i\in N,$$
\begin{equation}\text{and }D(x_i,x_{-i})=x_i D(x_i=1,x_{-i})+(1-x_i)D(x_i=0,x_{-i})\text{, for every }i\in N.\label{eq:mix}\end{equation}
\end{enumerate}
\end{assumption}

The third assumption imposes that deliberation protocols are deterministic, that is, if all team-members use a pure strategy, so that  $x_i\in\{0,1\}$ for all $i\in N$, then the team also uses a pure disclosure strategy. This means that there is no stochasticity inherent to the deliberation process. However, the team may still us mixed disclosure strategies, whenever those arise from team-members themselves mixing on their individual disclosure decisions. Indeed, condition (\ref{eq:mix}) states that if individual $i$ favors disclosure with probability $x_i$, then  the team's strategy will be consistent with $i$ using strategy $1$ with probability $x_i$ and strategy $0$ with probability $1-x_i$.\footnote{Part 3 of Assumption \ref{as:1} can be weakened to the following: if $D(x_i=1,x_{-i})>D(x_i=0,x_{-i})$, then $D(x_i,x_{-i})$ is strictly increasing in $x_i$ for all $x_i\in[0,1]$.} 
One special deliberation class is that of \emph{anonymous protocols}, in which disclosure decisions are made by ``$K$-majority rule:'' the team chooses disclosure if at least $K\leqslant N$ members vote for disclosure, and chooses no-disclosure otherwise. Formally, for every $I\subseteq N$, $D(1_{I},0_{-I})=1$ if $|I|\geqslant K$ and $D(1_{I},0_{-I})=0$ otherwise.

 Upon seeing the disclosed/not-disclosed outcome, the outside observer forms a posterior about the outcome that led to that observation. The observer's mean posterior about the realized outcome is equal to $\omega$, if $\omega$ is disclosed; and if the outcome is not disclosed, it is given by
\begin{align}
\label{eq:1}
\omega^{ND}_i\equiv\mathbb{E}\left[\omega_i|\text{no disclosure}\right]=
\frac{\int_\Omega \omega_i(1-d(\omega))dF(\omega)}{\int_\Omega (1-d(\omega))dF(\omega)},
\end{align}
for each $i\in N$, if $\int_{\Omega} (1-d(\omega))dF(\omega)>0$. If no disclosure is an off-path (measure zero) event, then the observer's mean posterior is indeterminate. Agents' payoffs are increasing in the observer's belief about their outcomes. Specifically, for each $i\in N$, agent $i$'s payoff if outcome $\omega$ is disclosed is $\omega_i$.
Their payoff is $\omega_i^{ND}$ if instead the outcome is not disclosed.\footnote{We assume that agents' payoffs depend on the posterior \emph{mean} induced on the observer about their outcomes. However, we can also allow agents to value other moments of the outcome distribution, by renormalizing the outcomes. For example, if agents' payoffs are given by $\mathbb{E}(\omega_i^2)$, they can be equivalently expressed by $\mathbb{E}(\nu_i)$, where $\nu_i=\omega_i^2$. In that case, we would take $F$ to be the joint distribution of $(\nu_1,...,\nu_N)$.}
\subsection{Equilibrium}
\label{sec:eqdisc}
Throughout the paper, for any subgroup $I\subseteq N$, we use the notation $(1_{I},0_{-I})$ to indicate a decision vector where each member of subgroup $I$ is in favor of disclosure and each member of subgroup $-I=N\setminus I$ opposes disclosure.
\begin{definition}\label{def:eq}
Given a deliberation process $D$, an \textbf{equilibrium} is an individual disclosure strategy $x_i:\Omega_i\rightarrow[0,1]$ for each $i\in N$, and no-disclosure posteriors $\omega^{ND}=(\omega_1^{ND},..., \omega_N^{ND})$  that satisfy
\begin{enumerate}
\item No individual or coalitional deviations:

For each $I\subset N$, $D(1_{I},x_{-I}(\omega))>D(0_{I},x_{-I}(\omega))$ implies
$$\omega_i> \omega^{ND}_i\hskip5pt\forall i\in I\Rightarrow x_i(\omega)=1\hskip5pt\forall i\in I,$$
$$\text{and }\omega_i< \omega^{ND}_i\hskip5pt\forall i\in I\Rightarrow x_i(\omega)=0\hskip5pt\forall i\in I.$$
\item The team's disclosure decision aggregates individual decisions: 
$$d(\omega)=D(x(\omega))\text{, for each }\omega\in\Omega.$$
\item No-disclosure posteriors are Bayes-consistent: for each $i\in N$, $\omega^{ND}_i$ satisfies (\ref{eq:1}).
\end{enumerate}
\end{definition}

The equilibrium notion described above is Perfect Bayesian Equilibrium, in which all team members and the outside observer understand the deliberation process, and with the additional restriction that there are no coalitional deviations. This extra requirement is reflected in our condition 1: given the deliberation protocol, any \emph{pivotal subgroup} of agents in the team --- who have the power to affect the team's disclosure choice --- uses this power to disclose outcomes that strictly benefit the whole subgroup, and to not disclose outcomes that imply strict losses to the whole subgroup. By requiring that there be no coalitional deviations in equilibrium, we refine out equilibria where individuals vote for/against disclosure solely because they believe their vote is not pivotal, and the equilibrium strategies imply that their votes are not pivotal. 
Condition 2 states that the teams' equilibrium disclosure strategy is reached by aggregating the individual equilibrium disclosure strategies, according to the given deliberation process. Finally, condition 3 imposes Bayes-consistency for beliefs reached on the equilibrium path.

We know from previous literature on disclosure with verifiable information --- for a survey, see Milgrom (2008) --- that when disclosure decisions are made by a single individual, the unique equilibrium involves full disclosure of all outcomes. That is, any realized evidence is disclosed to the outside observer with probability $1$. The key insight to that result is that if the observer knows that an individual holds some hard evidence of their outcome, then the non-disclosure of that evidence makes the observer skeptical about the outcome realization. The observer's skepticism then generates an \emph{unraveling} of any equilibrium with (partial) non-disclosure. We first remark in passing that in our environment, if all team members have perfectly correlated outcomes, then the team-disclosure game is equivalent to a disclosure problem for a single individual (regardless of the deliberation procedure). 

\begin{obs}
Suppose $F$ is such that outcomes are perfectly correlated across team members. Then for any deliberation protocol $D$, the unique equilibrium outcome is full disclosure. 
\label{obs:1}
\end{obs}

To highlight the differences between individual- and team-disclosure problems, for the rest of the paper we assume that $F$ is such that team members outcomes are ``not too correlated.'' 

\begin{assumption}
The outcome distribution $F$ has a \emph{product support}, that is, $\Omega=\Omega_1\times...\times\Omega_N$ where $\Omega_i\subset \mathbb{R}$ has at least $2$ elements for all $i\in N$, and $F$ has full support over $\Omega$. 
\end{assumption}
Given this assumption, our results show that when disclosure decisions are made by teams, equilibria often involve partial non-disclosure --- in contrast with individual-disclosure environments.
We will see that the usual unraveling argument often fails because the outside observer cannot fully attribute a non-disclosure decision to a specific team member.

\subsection{Model Interpretation}\label{sec:interp}
A realization of the team outcome is a piece of evidence which conveys to the outside observer some information about the team-members, or about a state relevant to the team-members. For example, an entrepreneurial team may have worked on a startup project which is now ready for its initial launch. The launch of this new company informs current and future investors about the underlying merits of the team-members as entrepreneurs. In that scenario, the ``realized'' startup project $\omega$ is a (perhaps imperfect) physical piece of evidence which is informative about the founders' underlying types. And each dimension $\omega_i$ of that realized evidence should be interpreted as the reputational implication of evidence $\omega$ about team-member $i$; or more formally, the value to team-member $i$ of the posterior formed by the observer about they entrepreneurial merits after seeing realization $\omega$. Despite the startup co-founders launching a single company, it may have different reputation implications for each of them, and individual outcomes may therefore not be perfectly correlated. For example, one team-member may be the more technical cofounder, while another is responsible for marketing; and therefore a technically sound but ``badly packaged'' company would reflect well on the former member and poorly on the latter.

Another relevant scenario may be one in which there is a true state of the world that the observer wishes to learn about, and this observer is consulting with a team of experts. For example, say the observer is a politician in power who wishes to learn the best location for a new power plant. The team consists of a set of experts, who produce a report which is a piece of evidence  $\omega$ conveying the optimal location of the plant. However, these team-members are themselves residents of the area and have different preferences over the best location of the plant --- for instance, they may prefer to have the plant be located far from their homes. In that case, $\omega_i$ should be interpreted as the value to expert $i$ of having the power plant located at the ``optimal location'' implied by evidence $\omega$. Under this interpretation, individual outcomes are not perfectly correlated with each other because experts may live in different areas, and therefore have different preferences regarding the plant location.

These two scenarios exemplify the flexibility of the interpretation of outcomes and of the meaning of evidence in our model. The value of the outcome to a team member may be a stand-in for their continuation value in a longer game --- in the spirit of career concerns and reputation models --- or it may represent the team-member's preference over the observer's belief about an underlying state of the world. Our main assumption is that each team-member has preferences not about the true realized $\omega$, but rather about the observer's belief about that realized evidence.

Apart from the outcome and evidence structure, the main object we introduce in this model is the deliberation protocol which aggregates individual team-members' decisions. In a real-world scenario, deliberation is a perhaps lengthy process made up of formal rules and communication between team-members which somehow aggregates the interests of the group into a team decision. Indeed, previous literature --- such as Gerardi and Yariv (2007) --- highlights the interplay of formal rules and communication in shaping equilibrium behavior in a deliberative committee. In this paper, we interpret our deliberation protocol $D$ as a \emph{reduced form aggregation rule} which already accounts for that interplay and informs how individual recommendations map into a team decision.

Relatedly, our model assumes that every team-member $i$ uses an individual disclosure strategy $x_i$ which depends only on their own outcome $\omega_i$ --- as opposed to the full team-outcome vector $\omega$. 
A direct interpretation of this assumption is that each team-member only sees their own outcome before deciding on disclosure, and therefore cannot condition their action on other individuals' outcomes. We make two more observations regarding this assumption. First, any equilibrium under the assumption that $x_i$ depends only on $\omega_i$ (for each $i\in N$) is also an equilibrium in the environment where that is not required. Second, if we assumed that the outcome distribution has no mass points, then the assumption would be without loss, as there would be no other equilibria. However, our benchmark studies a finite-outcome environment --- in Appendix \ref{app:B1}, we re-state our results for an environment with continuous outcomes. 

In our finite-outcomes environment, the assumption that $x_i$ depends only on $\omega_i$ refines  the equilibrium set, by ensuring that agents do not coordinate their individual disclosure strategies for $\omega$ realizations in which one or many of them are indifferent between disclosure and no disclosure. In Appendix \ref{app:B2}, we show an example of one such equilibrium in which two team-members coordinate their disclosure recommendations. We also show that such equilibrium can also be achieved as an equilibrium under a different deliberation protocol and in which every agent conditions their individual strategies only on their own outcome. From that lens, an interpretation of the deliberation protocol is that it is the aggregation rule that already encompasses all the ``coordination'' between the team-members.

\section{Equilibrium Team Disclosure}\label{sec:eqd}
 In a team context, we define three classes of equilibria, accommodating the fact that an equilibrium may involve full-disclosure of some team-members' outcomes, but not of others'.

\begin{definition} We say that
\begin{enumerate}[i]
\item An equilibrium has \textbf{full disclosure} if the observer can always perfectly infer the realized outcome $\omega\in\Omega$ on path. Or, equivalently, if there is at most one $\omega\in\Omega$ such that $d(\omega)<1$.
\item An equilibrium has \textbf{partial disclosure} if it does not have full disclosure.
\item A partial disclosure equilibrium is \textbf{interior} if, upon ``seeing'' no disclosure, the observer cannot perfectly infer $\omega_i$ for any team-member $i\in N$. That is, for each team-member $i\in N$, there exist outcome realizations $\omega,\omega'\in\Omega$, with $\omega_i\neq\omega'_i$, such that $d(\omega),d(\omega')<1$.
\end{enumerate}
\end{definition}
Our initial result, Theorem \ref{th:1}, formalizes our statement that equilibria ``often'' involve partial disclosure. To that end, we say that the disclosure protocol $D$ is such that disclosure \emph{can be chosen unilaterally} by team member $i\in N$ if $D(1_{i},0_{-i})=1$. And we say that disclosure \emph{cannot} be chosen unilaterally if $D(1_{i},0_{-i})=0$ for all $i\in N$.

\begin{figure}[t]
     \begin{subfigure}[b]{0.3\textwidth}
    \resizebox{\linewidth}{!}{ \begin{tikzpicture}[very thick]
    \node[fill=lightgray!70,draw=black, minimum size=2.2cm, inner sep=0pt] (as) {$(\omega_1^L,\omega_2^L)$};
    \node[draw=black, minimum size=2.2cm, inner sep=0pt, above=-\pgflinewidth of as] (abh) {$(\omega_1^L,\omega_2^H)$};
    \node[draw=black, minimum size=2.2cm, inner sep=0pt, right=-\pgflinewidth of as] (abv) {$(\omega_1^H,\omega_2^L)$};
    \node[draw=black, minimum size=2.2cm, inner sep=0pt,   above right=-\pgflinewidth and -\pgflinewidth of as] {$(\omega_1^H,\omega_2^H)$};
\end{tikzpicture}}
     \centering
         \caption{Full disclosure}
         \label{fig:exFull}
     \end{subfigure}
     \hfill
     \begin{subfigure}[b]{0.3\textwidth}
    \resizebox{\linewidth}{!}{ \begin{tikzpicture}[very thick]
    \node[fill=lightgray!70,draw=black, minimum size=2.2cm, inner sep=0pt] (as) {$(\omega_1^L,\omega_2^L)$};
    \node[fill=lightgray!70,draw=black, minimum size=2.2cm, inner sep=0pt, above=-\pgflinewidth of as] (abh) {$(\omega_1^H,\omega_2^L)$};
    \node[draw=black, minimum size=2.2cm, inner sep=0pt, right=-\pgflinewidth of as] (abv) {$(\omega_1^L,\omega_2^H)$};
    \node[draw=black, minimum size=2.2cm, inner sep=0pt,   above right=-\pgflinewidth and -\pgflinewidth of as]{$(\omega_1^H,\omega_2^H)$};
\end{tikzpicture}}
     \centering
         \caption{Partial disclosure}
         \label{fig:exPartial}
     \end{subfigure}
     \hfill
     \begin{subfigure}[b]{0.3\textwidth}
    \resizebox{\linewidth}{!}{ \begin{tikzpicture}[very thick]
    \node[fill=lightgray!70,draw=black, minimum size=2.2cm, inner sep=0pt] (as) {$(\omega_1^L,\omega_2^L)$};
    \node[fill=lightgray!70,draw=black, minimum size=2.2cm, inner sep=0pt, above=-\pgflinewidth of as] (abh) {$(\omega_1^H,\omega_2^L)$};
    \node[fill=lightgray!70,draw=black, minimum size=2.2cm, inner sep=0pt, right=-\pgflinewidth of as] (abv) {$(\omega_1^L,\omega_2^H)$};
    \node[draw=black, minimum size=2.2cm, inner sep=0pt,   above right=-\pgflinewidth and -\pgflinewidth of as]{$(\omega_1^H,\omega_2^H)$};
\end{tikzpicture}}
\centering
         \caption{Interior equilibrium}
         \label{fig:exInterior}
     \end{subfigure}
      \caption*{{\footnotesize Possible equilibrium outcomes in an environment with two players, and where outcomes are binary: $\Omega_i=(\omega_i^L,\omega_i^H)$, for $i\in \{1,2\}$. In each panel, light gray coloring represents outcomes that are not disclosed in equilibrium.}}
           \caption{Three Equilibrium Types}
        \label{fig:1}
\end{figure}

\begin{theorem}
\label{th:1}
The following statements are true about the equilibrium set: \footnote{\label{ft:2}This equilibrium characterization holds in our model where there are finitely many outcomes. We note that if outcome distributions are continuous, then partial disclosure equilibria are always interior. In Appendix \ref{app:B1}, we state a version of Theorem \ref{th:1} that holds in the continuous-outcomes variation of the model.}
\begin{enumerate}
\item A full-disclosure equilibrium exists.
\item A partial-disclosure equilibrium exists if and only if disclosure cannot be chosen unilaterally by all team members. 
\item If disclosure cannot be chosen unilaterally, then every partial-disclosure equilibrium is interior. Conversely, if disclosure can be chosen unilaterally by some team member, then no partial-disclosure equilibrium is interior.
\end{enumerate}
\end{theorem}

Before discussing the result, we state a proposition characterizing individual strategies in equilibrium. To that end, we say that two equilibria are \emph{equivalent} if they induce the same team disclosure rules. We also use Proposition \ref{pr:1a} as a lemma towards the proof of Theorem \ref{th:1}. 

\begin{proposition}
\label{pr:1a}
Any equilibrium is equivalent to an equilibrium in which every team member uses a threshold individual disclosure strategy, namely,
\begin{equation}
\label{eq:2}
\text{For every }i\in N\text{, } \omega_i>\omega^{ND}_i\Rightarrow x_i(\omega)=1\text{ and }\omega_i<\omega^{ND}_i\Rightarrow x_i(\omega)=0.
\end{equation}
\end{proposition}
Full proofs of Proposition \ref{pr:1a} and Theorem \ref{th:1} are available in the Appendix. We first argue that a full-disclosure equilibrium always exists, where every individual in a team always votes for disclosure, because they believe the observer's no-disclosure posterior about their outcome to be ``very low.'' In turn, because no-disclosure only happens off-path, ``very low'' no-disclosure posteriors for every individual are Bayes-consistent. Though there is a full-disclosure equilibrium, we then proceed to argue that it may coexist with partial-disclosure, and interior, equilibria --- as described in statements 2 and 3 of the Theorem. 

To that end, we introduce a map $\Phi$ relating each ``candidate vector'' of equilibrium no-disclosure posteriors into vectors of ``individually rational'' no-disclosure posteriors that are consistent with the starting candidate vector, as follows. We start by positing an equilibrium in which team members conjecture that the observer's no-disclosure posterior mean about their outcome is given by some vector $\omega^{ND}$ in the convex hull of $\Omega$. Given these conjectures, each team member uses an individual disclosure strategy where they vote for disclosure if their realized outcome is better than the conjectured no-disclosure posterior about their own outcome --- as in (\ref{eq:2}). We then use these individual strategies, along with the aggregating deliberation procedure $D$ to calculate their implied Bayes-consistent no-disclosure posteriors, $\hat{\omega}^{ND}$. The map $\Phi$ summarizes this procedure, letting $\hat{\omega}^{ND}\in\Phi(\omega^{ND})$.\footnote{The map $\Phi$ is set-valued, because for each candidate vector of no-disclosure beliefs, there is a continuum of ``threshold strategies'' that may be used by each agent. Specifically, we allow each individual to use a mixed disclosure strategy when they are indifferent between the realized outcome and the conjectured no-disclosure value. We use Kakutani's fixed point theorem to argue that $\Phi$ has at least one fixed point, which defines a team-disclosure equilibrium.} It is easy to see that any fixed point of $\Phi$ defines an equilibrium of the team-disclosure game; and inn the Appendix, we show that $\Phi$ has at least one such fixed point.

Now suppose disclosure cannot be chosen unilaterally by all team members; specifically, suppose $i\in N$ cannot choose disclosure unilaterally. Then it must be that any fixed point of $\Phi$ is such that $\min(\Omega_i)<\omega_i^{ND}$, and therefore a partial disclosure equilibrium exists. To see why, note that in any ``candidate equilibrium,'' all agents must vote against disclosure if their worst possible outcome realizes. Therefore, if $i$ draws an outcome $\omega_i\in \Omega_i$ and all other team members draw their worst possible outcome ($\min(\Omega_j)$, for each $j\neq i$), then that team outcome is not disclosed --- remember that $i$ cannot unilaterally choose to disclose it. And so it must be that every $\omega_i\in\Omega_i$ is not disclosed with positive probability, and therefore $\Phi(\omega^{ND})\in int(co (\Omega))$ for every conjectured $\omega^{ND}$. Consequently, any fixed point of $\Phi$ must have $\min(\Omega_i)<\omega_i^{ND}$. With an analogous argument, it is easy to see that when disclosure cannot be chosen unilaterally by any team member, then all fixed points of $\Phi$ must be interior, where $\min(\Omega_i)<\omega_i^{ND}$ for all $i\in N$. 
Conversely,  if disclosure \emph{can} be chosen unilaterally, then no interior equilibrium can exist. We show this with a variation of the unraveling argument: if team member $i\in N$ \emph{can} choose disclosure unilaterally, then no equilibrium in which $\omega_i^{ND}>\min(\Omega_i)$ can be sustained. 

Figure \ref{fig:1} illustrates possible equilibrium outcomes for a team with two members. When a team has only two members, then there are three possible types of deliberation protocols: (i) both agents can unilaterally choose disclosure, that is, $D(0,1)=D(1,0)=1$; (ii) one agent can unilaterally choose disclosure, so that exactly one of $D(0,1)=1$ or $D(1,0)=1$ holds; and (iii) neither agent can unilaterally choose disclosure, that is, $D(0,1)=D(1,0)=0$. Panels (a), (b), and (c) respectively show team-disclosure equilibria under these three deliberation protocols, thereby illustrating the possibilities delineated in Theorem \ref{th:1}. When a partial-disclosure equilibrium exists, it is such that disclosure happens when sufficiently many team members draw sufficiently good outcomes (outcomes that are better than the observer's no-disclosure posterior). This equilibrium characterization resembles a team-disclosure version Dye's (1985) \emph{sanitization} equilibrium, in which an individual discloses outcomes if and only if they are ``good enough.'' In our team-disclosure environment, the deliberation protocol determines which subgroups of team members are ``sufficiently many'' to determine that an outcome be disclosed.

The equilibrium characterization is also reminiscent of that in Martini's (2018) model of multi-dimensional disclosure by a single sender. Martini (2018) shows that if a single sender separably values the receiver's posterior about each dimensional of the state, then partial-disclosure equilibria may exist if the sender's preferences are sufficiently convex. Such equilibria are supported by the fact that, upon seeing no disclosure, the receiver cannot distinguish on which dimension the sender drew ``bad news.'' In contrast with Martini's (2018) approach, we interpret each dimension of the state as being relevant to one of the many senders in a team, and study the equilibrium characterization implied by various deliberation procedures through which disclosure decisions are made.

Having highlighted in Theorem \ref{th:1} that partial-disclosure and interior equilibria often coexist with a full disclosure equilibrium, we state in Theorem \ref{th:2} a result refining the equilibrium set. To do so, we propose a criterion to delineate circumstances where full-disclosure equilibria are  plausible. By definition, these are equilibria such that no-disclosure does not happen on the path of play,\footnote{Or perhaps only when all team members draw their worst-possible outcome.} and therefore no-disclosure posteriors are not required to be Bayes-consistent. As usual with ``forward induction'' refinements, we wish to evaluate whether these off-path posteriors which support the full-disclosure equilibrium can be justified by some ``plausible'' off-path behavior. In the context of team disclosure, we posit that even off-path beliefs should be justified by some behavior that is consistent with the team's deliberation protocol. 
\begin{definition}
\label{def:cons}
No-disclosure posteriors $\omega^{ND}=(\omega_1^{ND},..., \omega_N^{ND})$ are \textbf{consistent with deliberation} for protocol $D$ if there exists some team disclosure decision $d:\Omega\rightarrow[0,1]$ with $d(\omega)<1$ for some $\omega$, and  individual disclosure strategies $x_i:\Omega\rightarrow[0,1]$ for each $i\in N$ such that
\begin{enumerate}
\item For each $i,j \in N$ with $j\neq i$, $x_i(\omega)$ is constant with respect to $\omega_j$.
\item The team's disclosure decision aggregates individual decisions: 
$$d(\omega)=D(x(\omega))\text{, for each }\omega\in\Omega.$$
\item No-disclosure posteriors are Bayes-consistent: for each $i\in N$, $\omega^{ND}_i$ satisfies (\ref{eq:1}).
\end{enumerate}
\end{definition}
The definition states that a vector of no-disclosure posteriors is consistent with deliberation if there is some set of individual disclosure strategies that imply that no-disclosure happens with positive probability (given the deliberation protocol), and such that the no-disclosure posteriors are Bayes-plausible given these strategies. Theorem \ref{th:2} shows a condition on the deliberation protocol that is necessary and sufficient for full-disclosure equilibria to be consistent with deliberation. Informally, this condition requires the deliberation protocol to be such that a decision to disclose is easier to reach than a decision to conceal an outcome.
Formally, we say that \emph{disclosing requires more consensus than concealing} if for every \emph{pivotal subgroup} $I\subseteq N$, such that $D(1_{I},0_{-I})=1$ and $D(0_{I},1_{-I})=0$, there exists a smaller subgroup $J\subset I$ such that $D(0_{J},1_{-J})=0$. For example, if we consider a team of only two individuals, then disclosure requires more consensus than concealing if and only if disclosure is decided consensually between the two team members.

\begin{theorem}
\label{th:2} A full-disclosure equilibrium that is consistent with deliberation exists if and only if disclosing does not  require more consensus than concealing. 

\end{theorem}

To understand the result, suppose there are only two team members ($N=2$). And suppose disclosure requires more consensus than concealing (that is, $D(0,1)=D(1,0)=0$), so that each team member can unilaterally choose to conceal a realization. Now take a pair individual disclosure strategies $x=(x_1,x_2)$ and a pair of no-disclosure beliefs $\omega^{ND}=(\omega^{ND}_1,\omega^{ND}_2)$ that constitute a full-disclosure equilibrium. It must be that $\omega^{ND}_1=\min(\Omega_1)$ and $\omega^{ND}_2=\min(\Omega_2)$; for otherwise one of the team members would strictly prefer to not disclose realizations where they draw their worst possible outcome, and they would be able to unilaterally impose such non-disclosure. This would contradict the initial assumption that the equilibrium has full-disclosure.

Now we wish to craft a pair of individual disclosure strategies $\hat{x}$, where each individual’s strategy depends only
on their own realized outcome, to be used to ``justify'' the off-path beliefs $(\omega^{ND}_1,\omega^{ND}_2)=(\min(\Omega_1),\min(\Omega_2))$; as in Definition \ref{def:cons}. These strategies must imply that some realization $\hat{\omega}$ is not disclosed with positive probability, and therefore it must be that either $\hat{x}_1(\hat{\omega}_1,\omega_2)<1$ for all $\omega_2\in\Omega_2$ or $\hat{x}_2(\omega_1,\hat{\omega}_2)<1$ for all $\omega_1\in\Omega_1$. If the former is true, then all realizations $\omega_2\in\Omega_2$ are concealed with positive probability, which implies that the no-disclosure posterior $\omega^{ND}_2$ consistent with $\hat{x}$ is strictly larger than $\min(\Omega_2)$. If the latter is true, then $\omega^{ND}_1>\min(\Omega_1)$. Combining these two cases, we conclude that the off-path beliefs necessary to sustain full-disclosure  cannot be justified by \emph{any} disclosure strategies consistent with the deliberation process; and therefore full-disclosure is not consistent with deliberation. 

With some work shown in the Appendix, this argument generalizes to teams with $N>2$ members, so long as the deliberation process is such that disclosing requires more consensus than concealing. More precisely, in any full-disclosure equilibrium there must be a subgroup $I\subseteq N$ of the team, who can together choose disclosure (that is, $D(1_{I},0_{-I})=1$) and such that $\omega^{ND}_i=\min(\Omega_i)$ for all $i\in I$. But we show that it is impossible to construct a strategy profile $\hat{x}$ that justifies these off-path beliefs. To argue this point, we use the fact that there is a subset of team members $J\subset I$ that can together choose no disclosure, that is, $D(0_{J},1_{-J})=0$.

To see the other direction of Theorem \ref{th:2}, let's again consider that there are only two team members, and now suppose that disclosure does not require more consensus than concealing. That is, either $D(0,1)=1$ or $D(1,0)=1$ --- suppose the former is true for the sake of this argument. Then there exists a full disclosure equilibrium where $\omega^{ND}_1=\min(\Omega_1)$ and $\omega^{ND}_2>\min(\Omega_2)$. Moreover, these off-path beliefs can be justified by the following individual disclosure strategies: $\hat{x}_1(\omega)=0$ if $\omega_1=\min(\Omega_1)$ and $\hat{x}_1(\omega)=1$ otherwise; and $\hat{x}_2(\omega)=0$ for all $\omega\in\Omega$. Therefore there exists a full-disclosure equilibrium that is consistent with deliberation. Once again, we show in the Appendix that this argument can be generalized to larger teams.

Together, Theorems \ref{th:1} and \ref{th:2} characterize how equilibrium disclosure ``decreases'' after an increase in the degree of consensus required for the team to disclose. In the first result, we see that unless disclosure is very easy --- in the sense that it can be chosen unilaterally by a team member --- then full-disclosure is not the unique equilibrium outcome. And further, if no team-member can choose to disclose unilaterally, then interior equilibria exist in which no team member has their outcomes fully revealed. Theorem \ref{th:2} is even stronger, delineating a necessary and sufficient condition under which not only a partial-disclosure equilibrium exists, but also it is more plausible than full-disclosure. This condition is that non-disclosure is more easily attainable than disclosure. 

The predictions made by Theorems \ref{th:1} and \ref{th:2} are particularly clear when we consider anonymous deliberation processes. Remember that when deliberation processes are anonymous, disclosure decisions are made by ``$K$-majority rules'': the team chooses disclosure if at least $K\leqslant N$ members vote for disclosure, and chooses no-disclosure otherwise. Within this class, disclosure can be chosen unilaterally if and only if $K=1$; and disclosure requires more consensus than concealing if and only if $K>N/2$.

\begin{corollary}[to Theorems \ref{th:1} and \ref{th:2}]
Suppose $D$ is an $n$-majority deliberation procedure, for $n\leqslant N$. Full-disclosure is the unique equilibrium outcome if and only if $K=1$; and all equilibria that are consistent with deliberation are interior if and only if $K\leqslant N/2$.
\end{corollary}

\subsection{Equilibrium Team Disclosure with Binary Outcomes}
\label{sec:binary}
To further characterize equilibrium team-disclosure outcomes, and their relation to the deliberation procedure, we now consider an environment where each team member draws one of two outcomes. 

\begin{definition}
The distribution $F$ has \textbf{binary outcomes} if for each $i\in N$,
$$\Omega_i=\{\ell_i,h_i\}\text{, with }\ell_i<h_i.$$
If $F$ has binary outcomes, then for each $\omega\in\Omega$, we let 
$$H(\omega)=\{i\in N:\omega_i=h_i\},$$
$$\text{and }L(\omega)=N\setminus H(\omega)=\{i\in N:\omega_i=\ell_i\}.$$
\end{definition}
Proposition \ref{pr:1} characterizes interior equilibria  in the binary outcome setting. Remember from Theorem \ref{th:1} that these interior equilibria exist when the deliberation protocol is such that disclosure cannot be chosen unilaterally.
\begin{proposition}
\label{pr:1}
Suppose $F$ has binary outcomes. Then in any interior equilibrium, the team disclosure strategy is given by $d(\omega)=D(1_{H(\omega)},0_{L(\omega)})$.
\end{proposition}

For any deliberation protocol $D$, the interior equilibrium described in Proposition \ref{pr:1} is the equilibrium with least disclosure that can be sustained with that protocol. We now state corollaries to Proposition \ref{pr:1}, performing comparative statics relating equilibrium disclosure/no-disclosure to the deliberation protocol and to the correlation of outcomes between team-members. Corollary \ref{cor:2} shows that making disclosure easier --- by decreasing the degree of consensus required for the team to choose disclosure --- leads to more disclosure in equilibrium.
\begin{corollary}[to Proposition \ref{pr:1}]
Suppose $F$ has binary outcomes, and take two deliberation protocols $D$ and $D'$ such that 
$$D(x)=1\Rightarrow D'(x)=1.$$
If $d$ is an equilibrium team-dislcosure strategy under protocol $D$, then there is some $d'$ which is an equilibrium team-disclosure strategy under protocol $D'$ such that 
$$d'(\omega)\geqslant d(\omega)\text{, for every }\omega\in\Omega.$$
\label{cor:2}
\end{corollary}
Next, Corollary \ref{cor:3} shows that when the correlation between team-members' outcomes increases, equilibrium disclosure becomes more correlated with each agent's outcome --- in the sense that the disclosure of an agent's high outcome becomes more likely and the disclosure of the same agent's low outcome becomes less likely. To state this result, we propose a definition of ``increasing the correlation'' between the individual outcomes of the team members. Suppose $F$ and $F'$ are two outcome distributions with the same marginal distributions over $\Omega_i=\{\ell_i,h_i\}$ for every $i\in N$. We say $F'$ features \emph{more outcome correlation} than $F$ if
$$F'(\omega_i=\ell_i|\omega_j=\ell_j)\geqslant F(\omega_i=\ell_i|\omega_j=\ell_j)\text{ for every }i,j\in N,$$
$$\text{ and }F'(\omega_i=h_i|\omega_j=h_j)\geqslant F(\omega_i=h_i|\omega_j=h_j)\text{ for every }i,j\in N.$$
\begin{corollary}
Suppose $F$ and $F'$ are two outcome distributions with the same marginal distributions over $\Omega_i=\{\ell_i,h_i\}$ for every $i\in N$; and suppose $F'$ features more outcome correlation than $F$. If $d$ is the interior equilibrium team-dislcosure strategy under $F$, then it is the interior equilibrium team-disclosure strategy under $F'$, and the following hold: 
$$\text{1. }\mathbb{P}_{F',d}(d=1,\omega_i=h_i)\geqslant \mathbb{P}_{F,d}(d=1,\omega_i=h_i)\text{ for every }i\in N,$$
$$\text{2. }\mathbb{P}_{F',d}(d=0,\omega_i=\ell_i)\geqslant \mathbb{P}_{F,d}(d=0,\omega_i=\ell_i)\text{, for every }i\in N.$$
\label{cor:3}
\end{corollary}
For this result, we consider increasing the correlation between agents' outcomes without ever making them perfectly correlated --- remember that we are maintaining the assumption of $\Omega$ being a product space throughout. However, suppose we consider a sequence of ``correlation increases'' that converges to fully correlated outcomes; that is, for all $i,j\in N$, $F(\omega_i=h_i|\omega_j=h_j), F(\omega_i=\ell_i|\omega_j=\ell_j)\rightarrow 1$. In the limit, the interior equilibrium considered in Corollary \ref{cor:3} would be such that $\mathbb{P}_{F',d}(d=1,\omega_i=h_i)\rightarrow 1$ and $\mathbb{P}_{F',d}(d=0,\omega_i=\ell_i)\rightarrow 1$, for any $i\in N$. In words, all agents' high outcomes would be disclosed with certainty and all agents low outcomes would be not disclosed, also with certainty. Note that such an equilibrium effectively features full-disclosure --- because the only outcome that is not disclosed is that in which all agents drew the low outcome. This means that, as team members' outcomes become close to fully correlated, the interior equilibrium effectively converges to the full disclosure equilibrium described in Observation \ref{obs:1}.

\section{Deliberation and Incentives}
\label{sec:inc}
In this section, we argue that the disclosure-deliberation protocol, through determining the team project outcomes that are disclosed or not disclosed to the outside observer, can have an effect on team-members' incentives to put effort into the team project. To study the effect of deliberation of effort incentives, we add a pre-disclosure stage to the team's problem. Formally, each agent $i\in N$ puts effort $e_i\in\{0,1\}$
into the team project, incurring in cost $c_i>0$ if $e_i=1$ and no cost otherwise. Given an effort vector $e=(e_1,...,e_N)$, the outcome distribution is $F(\cdot;e)$. We assume that the support of outcomes is invariant to the chosen vector of efforts, and that the outcome distribution increases in the teams effort.

\begin{assumption}
\label{as:3} For each $e\in\{0,1\}^N$, $F(\cdot;e)$ has full support over $\Omega=\Omega_1\times...\times\Omega_N$, where $\Omega_i\subset \mathbb{R}$ has at least $2$ elements for all $i\in N$. Moreover, effort is productive, so that\footnote{The notation $\succsim_{FOS}$ indicates (multivariate) first order stochastic dominance. We say that a random vector $X$ dominates a random vector $Y$ in the first order stochastic if $\mathbb{P}(X\in U)\geqslant \mathbb{P}(Y\in U)$ for every upper set $U\in\mathbb{R}^n$. Equivalently, random vector $X$ dominates random vector $Y$ in the first order stochastic if $\mathbb{E}\left[\varphi(X)\right]\geqslant \mathbb{E}\left[\varphi(Y)\right]$ for all increasing functions $\varphi$ for which the expectations exist. See Shaked and Shanthikumar (2007).} 
$$e\geqslant e'\Rightarrow F(\cdot;e)\succsim_{FOS} F(\cdot;e').$$
\end{assumption}
The timing of the game is as follows: (i) Each agent unilaterally makes their effort choice, yielding effort vector $e$; (ii) Team outcome realizes, drawn from $F(\cdot;e)$; (iii) All team members see the drawn outcome $\omega\in\Omega$ and make their individual disclosure decisions; (iv) Individual disclosure decisions are aggregated according to the deliberation process $D$; (v) Outcome $\omega$ is disclosed/not disclosed to the observer. As before, the deliberation process $D$ is common knowledge to all team-members and the observer. At stage (v), the observer sees the disclosed/not-disclosed outcome, but not the team-members' effort choices (that is, effort decisions are ``covert''). An equilibrium of this larger game is defined by an equilibrium of the team-disclosure game (as in Definition \ref{def:eq}) and individual rationality at the effort-choice stage given the team-disclosure equilibrium.

\subsection{Strategic Disclosure and Effort Incentives}
Throughout our analysis, we evaluate the incentives provided by different deliberation processes in terms of whether --- and for what cost vectors --- they can implement equilibria with full effort. 
Theorem \ref{th:3} below establishes the basis for this analysis, clarifying the relation between disclosure strategies implemented in the team-disclosure stage and team-members' incentives to exert costly effort. Let $c\in\mathbb{R}^N_{++}$ be the vector of effort costs for the team. 
For any team-disclosure strategy $d:\Omega\rightarrow[0,1]$, let $fe(d)\subset\mathbb{R}^N_{++}$ be its corresponding \emph{full-effort set}. That is, $c\in fe(d)$ if, given the team-disclosure strategy $d$, there is an equilibrium of the effort-choice stage in which $e_i=1$ for all $i\in N$. For any subgroup $I\subset N$, we use notation $e_I$ to indicate an effort vector such that individuals $i\in I$ exert effort and individuals $i\in N\setminus I$ do not.

\begin{theorem}
\label{th:3}
A team-disclosure strategy $d:\Omega\rightarrow [0,1]$ implements full effort for a given cost vector $c\in\mathbb{R}^N$ --- that is, $c\in fe(d)$ --- if and only if for every $i\in N$ \footnote{\label{ft:1}The following rewriting of (\ref{eq:p1}) expresses the relation with the team disclosure strategy $d$ more directly:
\begin{equation}c_i\leqslant \int_\Omega\omega_idF(\omega;e_N)-\int_\Omega\omega_idF(\omega;e_{N\setminus i})\nonumber\end{equation}
$$-\int_\Omega(1-d(\omega))dF(\omega;e_{N\setminus i})\left[\frac{\int_\Omega\omega_i(1-d(\omega))dF(\omega;e_N)}{\int_\Omega(1-d(\omega))dF(\omega;e_N)}-\frac{\int_\Omega\omega_i(1-d(\omega))dF(\omega;e_{N\setminus i})}{\int_\Omega(1-d(\omega))dF(\omega;e_{N\setminus i})}\right].$$}
\begin{equation}\mathbb{E}\left[\omega_i|e_N\right]-\mathbb{E}\left[\omega_i|e_{N\setminus i}\right]-\mathbb{P}\left[ND|e_{N\setminus i}\right]\bigg\{\mathbb{E}\left[\omega_i|ND; e_N\right]-\mathbb{E}\left[\omega_i|ND; e_{N\setminus i}\right]\bigg\}\geqslant c_i.\label{eq:p1}\end{equation}
\end{theorem}


The expression in (\ref{eq:p1}) clarifies how the selective disclosure of the teams outcomes can be used to incentivize team-members to put in effort beyond their baseline ``full-disclosure'' effort. On the left-hand side of (\ref{eq:p1}), the first two terms correspond to the difference between individual $i$'s expected value when they choose $e_i=1$ versus $e_i=0$ --- while maintaining the assumption that all other team-members exert effort. If agent $i$ foresees that all outcomes will be disclosed, then to make their effort choice they compare this difference (the expected gain from effort) to the cost of effort on the right-hand side. Indeed, note that if $d(\omega)=1$ for all $\omega\in\Omega$, then the probability of no-disclosure is equal to zero, and therefore so is the third term on the left-hand side of (\ref{eq:p1}).

If instead the team-disclosure strategy $d$ differs from full disclosure, then the third term in the left-hand side of (\ref{eq:p1}) measures the extra effort incentive provided by strategic disclosure. This term is positive whenever the expected outcome for agent $i$, \emph{conditional on non-disclosure}, decreases with effort. Note that to calculate these conditional expectations, we maintain the team-disclosure strategy $d$ unchanged,\footnote{In any Perfect Bayesian Equilibrium, equilibrium disclosure strategies in the disclosure stage must not depend on the effort choice in the initial stage, because effort is chosen covertly by each agent.} as a function of the realized outcome $\omega$, and vary the outcome distribution with $i$'s effort choice --- see footnote \ref{ft:1} for a direct expression. Intuitively, the non-disclosure expected outcome to agent $i$ decreases with $i$'s effort if by exerting effort, agent $i$ improves the correlation between their outcome and disclosure. Indeed, the following rewriting of the left-hand side of (\ref{eq:p1}) expresses $i$'s effort gains directly in terms of the improvement of the covariance between $i$'s outcome and disclosure:\footnote{Please see the Appendix, where we derive this expression from (\ref{eq:p1}).}
\begin{align}
&\big(1-\mathbb{P}\left[ND|e_{N\setminus i}\right]\big)\big(\mathbb{E}\left[\omega_i|e_N\right]-\mathbb{E}\left[\omega_i|e_{N\setminus i}\right]\big)\label{eq:q1}
\\[1em]
&\hskip100pt+\frac{\mathbb{P}\left[ND|e_{N\setminus i}\right]}{\mathbb{P}\left[ND|e_{N}\right]}Cov\left[\omega_i,d|e_N\right]-Cov\left[\omega_i,d|e_{N\setminus i}\right]\nonumber\end{align}
Expression (\ref{eq:q1}) clarifies that selective non-disclosure has two impacts on effort incentives. On the one hand, it increases the non-disclosure region (when compared to full disclosure), which decreases $1-\mathbb{P}\left[ND|e_{N\setminus i}\right]$, thereby negatively impacting the gains from effort. On the other hand, it may generate extra effort incentives for agent $i$ if the covariance between $\omega_i$ and disclosure $d$ is improved with $i$'s effort.

\subsection{Deliberation and Effort Incentives}
Theorem \ref{th:3} characterizes the effort incentives provided by different disclosure strategies. We now turn to the evaluation of the incentives provided by \emph{equilibrium} disclosure strategies implied by different deliberation procedures. For each deliberation process $D$,  $FE(D)\subset\mathbb{R}^N_{++}$ is its corresponding full-effort set. That is, $c\in FE(D)$ if, given the deliberation process $D$, there is a team-disclosure strategy $d$, with $c\in fe(d)$, that can be sustained in an equilibrium with the full-effort outcome distribution $F(\cdot;e_N)$.

Our first result ranks deliberation processes in terms of their effort incentives for two ``types of effort.'' We say effort is \emph{self-improving} if, for every $i\in N$ and every $I\subset N$,\footnote{The notation $\succ_{FOS}$ indicates \emph{strict} (multivariate) first order stochastic dominance. We say that a random vector $X$ strictly dominates a random vector $Y$, both defined over $\Omega$, in the first order stochastic if $\mathbb{P}(X\in U)> \mathbb{P}(Y\in U)$ for every upper set $U\in\Omega$.} 
$$F_{N\setminus i}(\cdot;e_I)=F_{N\setminus i}(\cdot;e_{I\setminus i}),\text{ and }F_i(\cdot|\omega_{N\setminus i};e_I)\succ_{FOS}F_i(\cdot|\omega_{N\setminus i};e_{I\setminus i}).$$
The notation $F_{N\setminus i}(\cdot;e_I)$ indicates the joint distribution of outcomes of all team members except team-member $i$, when the effort vector is $e_J$, that is, when effort is exerted by all agents in $I$ and no agents in $N\setminus I$. In turn, $F_i(\cdot|\omega_{N\setminus i};e_I)$ indicates the outcome distribution for team-member $i$ conditional on outcome realization $\omega_{N\setminus i}$ for all other team-members, given effort vector $e_J$. Accordingly, we say effort is self-improving if, for every team-member $i\in N$, their own effort leaves the outcome distribution of other team-members unchanged, but improves their own outcome distribution, conditional on others' outcome realization. Contrastingly, we say effort is \emph{team-improving} if for every $i\in N$ and every $I\subset N$, 
$$F_{N\setminus i}(\cdot|\omega_i;e_I)\succ_{FOS}F_{N\setminus i}(\cdot|\omega_i;e_{I\setminus i}),\text{ and }F_i(\cdot;e_I)=F_i(\cdot;e_{I\setminus i}).$$
If outcomes are independently drawn across team-members, then self-improving effort corresponds to a situation where $i$'s outcome distribution increases in the first order stochastic if $i$ exerts effort, and $j$'s outcome distribution remains unchanged for all $j\neq i$. Also in the independent case, team-improving effort is such that $j$'s outcome distribution increases with $i$'s effort, for $j\neq i$, but $i$'s distribution is unchanged. Proposition \ref{pr:3} provides three statements ranking deliberation protocols when effort is either self-improving or team-improving. To that end, we say that a deliberation protocol $D$ \emph{dominates} a deliberation protocol $D'$ if $FE(D')\subseteq FE(D)$. Additionally, we say $D$ is the \emph{unilateral disclosure} protocol if it is the deliberation protocol in which \emph{every} team-member can unilaterally choose disclosure; and we say it is the \emph{consensual disclosure} protocol if every team-member can unilaterally choose \emph{no}-disclosure.

\begin{proposition}
\label{pr:3}
\begin{enumerate}
\item If effort is \underline{self-improving}, then the unilateral disclosure protocol dominates any other disclosure protocol.
\item If effort is \underline{team-improving}, then the unilateral disclosure protocol is strictly dominated by the consensual disclosure protocol. 
\item If effort is \underline{team-improving}, then the consensual disclosure protocol is not dominated by any deliberation protocol in which disclosure can be chosen unilaterally by some team-member.
\end{enumerate}
\end{proposition}
A full proof of Proposition \ref{pr:3} is in the Appendix. When effort is self-improving, we show that all of $i$'s gains from effort are captured when the team fully discloses their outcomes, which is the (unique) equilibrium team-disclosure attained when the deliberation process allows any team-member to unilaterally choose disclosure. Intuitively, because an agent's effort affects only their own outcome, non-disclosure can only harm effort incentives by concealing some of the effort gains from the observer. Indeed, we show that as a consequence, the unilateral disclosure protocol --- by inducing an equilibrium with full disclosure --- maximizes the team's ``full effort cost set.'' We note that, as per Theorem \ref{th:1}, full-disclosure is an equilibrium for any deliberation procedure; and therefore the maximal full-effort cost set can be attained regardless of the deliberation process. However, by Theorem \ref{th:2}, we can refine the equilibrium set induced by different deliberation protocols --- specifically, we can refine out the full-disclosure equilibrium when ``disclosing requires more consensus than concealing.'' If we accordingly define our dominance criterion accounting for this refinement, we can then establish that the unilateral disclosure protocol strictly dominates deliberation protocols in which disclosing requires more consensus than concealing. 

Now suppose instead that effort is team-improving. Statement 2 in Proposition \ref{pr:3} argues that the equilibrium team-disclosure strategy implemented by the consensual disclosure deliberation protocol produces larger effort gains than full-disclosure (which is the unique disclosure equilibrium under the unilateral disclosure protocol). Consider the consensual disclosure protocol, and remember that an interior equilibrium exists in which each team-member favors disclosure if and only if their own-outcome draw is good-enough (where good enough is determined by some interior threshold). When team-member $i$ puts in effort, they improve the odds that all other team-members will draw an outcome for which they will favor disclosure; therefore improving the odds that $i$'s disclosure decision is pivotal. And finally, by increasing $i$'s pivotality, $i$'s effort then increases the correlation between the teams' disclosure choice and $i$'s own-outcome. As argued in Theorem \ref{th:3}, this improved correlation adds to $i$'s gains from effort --- the third term on the left-hand side of (\ref{eq:p1}) is positive. 

In other words, under the consensual disclosure deliberation protocol, each team-member has incentives to improve the outcomes of their partners, so as to avoid situations where the disclosure of their own good outcome realizations is vetoed by others. This feature of \emph{interior} team-disclosure equilibria creates strategic complementarities between team-members. Statement 3 in Proposition \ref{pr:3} follows the same logic just described for statement 2. However, instead of comparing consensual-disclosure only to unilateral-disclosure protocols, we consider other deliberation protocols where some team members can unilaterally disclose. If $i$ can unilaterally disclose, then any equilibrium of the team-disclosure game must fully disclose $i$'s outcomes --- and therefore $i$'s effort incentives equal those in the unilateral-disclosure protocol, thus implying the ranking in statement 3. Note that statement 3 says only that consensual disclosure is \emph{not dominated by} protocols where some members can unilaterally disclose; highlighting that our ordering is not complete.\footnote{As previously discussed in footnote \ref{ft:2} and stated in Appendix \ref{app:B1}, if outcomes are continuously distributed, then the unique disclosure equilibrium for any protocol where some team-member can choose disclosure unilaterally is full-disclosure. In that case, statement 3 in Proposition \ref{pr:3} can be strengthened to say that consensual disclosure \emph{strictly dominates} protocols where some members can unilaterally disclose. We formally state this in Appendix \ref{app:B1}.}

In Proposition \ref{pr:4} below, we consider a broader class of ``effort types'', and rank the effort incentives provided by deliberation protocols against those provided by the unilateral disclosure protocol. Specifically, we show that if each team-member's effort sufficiently improves the correlation between all team-member's outcomes, then all deliberation protocols dominate unilateral disclosure. To that end, we momentarily assume that the support of outcomes does not differ across agents, so that $\Omega=\Omega_i^N$ for some $\Omega_i\subset \mathbb{R}$; and we say that a distribution $G$ over $\Omega$ has perfect correlation across team-members' outcomes if it has full support on the locus $\omega_1=...=\omega_N$.\footnote{These assumptions are made for notational convenience. Proposition \ref{pr:4} holds even if the support of outcomes differs across agents, and under the weaker assumption that $G$ is supported on the locus $\omega_j=\varphi_{ij}(\omega_i)$ for some strictly increasing function $\varphi_{ij}$ for all $i,j\in N$.}

\begin{proposition}
\label{pr:4}
Suppose $F$ and $G$ are two distributions over $\Omega=\Omega_i^N$, where $F$ has full support and $G$ has perfect correlation across team-members' outcomes, and suppose $G\succsim_{FOS}F\succsim F(\cdot;e_{N\setminus i})$ for every $i\in N$. Consider varying the correlation in $F(\cdot;e_N)$ by letting, for $\epsilon\in(0,1)$,
$$F_\epsilon(\cdot;e_N)=(1-\epsilon)F+\epsilon G.$$
Let $D$ be the unilateral disclosure deliberation protocol and $D'$ be a deliberation protocol with $D'\neq D$. There exists some $\bar{\epsilon}\in(0,1)$ such that, if $\epsilon>\bar{\epsilon}$ and $F_\epsilon(\cdot;e_N)$ is the full-effort outcome distribution, then $D'$ strictly dominates $D$. 
\end{proposition}
Finally, Proposition \ref{pr:5} assesses the merits of nominating one of the team-members as a \emph{team leader}, in terms of effort incentives. A deliberation protocol has a team leader if there is some team-member $i\in N$ such that $D(x)=x_i$. And we say $i\in N$ is an \emph{effective} team leader if the deliberation protocol where $i$ is the leader strictly dominates the unilateral disclosure protocol. Proposition \ref{pr:5} is almost a corollary of Theorem \ref{th:3}, following from the observation that when $i\in N$ is a team-leader, then equilibrium team-disclosure strategies are either full disclosure, or conceal if and only if $i$'s outcome is $\ushort{\omega}_i$, their worst possible outcome realization. The proposition then states that $i$ is an effective team leader if, by exerting effort, every other team-member decreases their expected outcome conditional on $i$'s worst outcome --- in other words, $i$ is an effective leader if other team-member's outcomes become ``more correlated'' with $i$'s when they exert effort. 

\begin{proposition}
\label{pr:5}
Team-member $i\in N$ is an effective team leader if
$$\mathbb{E}\left[\omega_j|\ushort{\omega}_i; e_N\right]< \mathbb{E}\left[\omega_j|\ushort{\omega}_i; e_{N\setminus j}\right]\text{ for all }j\neq i.$$
\end{proposition}

\subsection{Deliberation and Incentives in a Binary Environment}
\label{sec:incbin}
To further characterize the relationship between disclosure protocols and the provision of effort in teams, we consider a symmetric environment with binary outcomes. In this example, individuals' outcome distributions are described as follows. With probability $p\in(0,1)$, all team members receive the same ``team outcome,'' so that $\omega_i=\hat{\omega}_T\in\{\ell_T,h_T\}=\{0,1\}$ for all $i\in N$, and with complementary probability $(1-p)$, each team-member $i\in N$ draws their own outcome $\omega_i=\hat{\omega}_i\in\{\ell_i,h_i\}=\{0,1\}$ independently. The distribution of team outcomes is described by $\mathbb{P}(h_T)$, and each team-member's own independent outcome distribution is described by $\mathbb{P}(h_i)$. In our exercise, we let individuals' efforts affect each of the parameters of the joint outcome distribution, so that $p$, $\mathbb{P}(h_T)$, and $\mathbb{P}(h_i)$ for each $i\in N$ are to be seen as functions of the team's chosen effort profile.


Proposition \ref{pr:eff_bin} studies the incentive properties of \emph{symmetric} deliberation procedures considering different possible effects of effort on the outcome distribution. Symmetric deliberation procedures are $K$-majority protocols, for some $1\leqslant K \leqslant N$, in which the team discloses the outcome if and only if at least $K$ team members favor disclosure.

\begin{proposition}
The unilateral disclosure protocol ($K=1$) is strictly dominated by all $K$-majority protocols with $K>1$ if
\begin{enumerate}
\item[(i)]    For every $i\in N$, $\mathbb{P}(h_j|e_N)>\mathbb{P}(h_j|e_{N\setminus i})$ for all $j\neq i$, $\mathbb{P}(h_i|e_N)=\mathbb{P}(h_i|e_{N\setminus i})$, $\mathbb{P}(h_T|e_N)=\mathbb{P}(h_T|e_{N\setminus i})$, and $p(e_N)=p(e_{N\setminus i})$. 

\item[(ii)]    For every $i\in N$,  $p(e_N)>p(e_{N\setminus i})$, $\mathbb{P}(h_j|e_N)=\mathbb{P}(h_j|e_{N\setminus i})$ for all $j\in N$, and $\mathbb{P}(h_T|e_N)=\mathbb{P}(h_T|e_{N\setminus i})$.
\end{enumerate}
The unilateral disclosure protocol ($K=1$) dominates all $K$-majority protocols if
\begin{enumerate}
    \item[(iii)]     For every $i\in N$, $\mathbb{P}(h_i|e_N)>\mathbb{P}(h_i|e_{N\setminus i})$, $\mathbb{P}(h_j|e_N)=\mathbb{P}(h_j|e_{N\setminus i})$ for all $j\neq i$, $\mathbb{P}(h_T|e_N)=\mathbb{P}(h_T|e_{N\setminus i})$, and $p(e_N)=p(e_{N\setminus i})$.
    
    \item[(iv)]    For every $i\in N$, $\mathbb{P}(h_T|e_N)>\mathbb{P}(h_T|e_{N\setminus i})$, $\mathbb{P}(h_j|e_N)=\mathbb{P}(h_j|e_{N\setminus i})$ for all $j\in N$, and $p(e_N)=p(e_{N\setminus i})$.
\end{enumerate}
\label{pr:eff_bin}
\end{proposition}


The proof of Proposition \ref{pr:eff_bin} is presented in the Appendix. It uses the tractability of the binary environment with a symmetric deliberation process to simplify the inequality obtained in Theorem \ref{th:3}. The first two statements in the proposition are stronger versions of statement 2 in Proposition \ref{pr:3} and of Proposition \ref{pr:4}, respectively. Statement (i) shows that, in this binary environment, if effort is team-improving, then the unilateral disclosure protocol is dominated by all symmetric deliberation protocols that require more consensus --- and in particular the consensual disclosure protocol. Statement (ii) shows that if effort affects the correlation between team-members' outcomes --- not necessarily to an extreme degree as in Proposition \ref{pr:4} --- then all symmetric deliberation protocols dominate the unilateral disclosure protocol. Statement (iii) is a direct application of the first statement in Proposition \ref{pr:3}, regarding self-improving effort. And finally, statement (iv) evaluates symmetric deliberation protocols when effort improves the chances of a high common outcome, thereby improving individuals' own outcome distributions as well as their team members'. 

The intuition for this last item is not fully captured by the logic behind  Propositions \ref{pr:3} and \ref{pr:4}. To understand it, note that  there are two different possible scenarios in which non-disclosure occurs. Either the whole team obtains the same low common outcome, or enough team members independently receive low outcomes. If effort increases the likelihood of a good common outcome then, conditional on non-disclosure, the second option becomes more likely. As in this case player $i$ might have had a  good outcome, we have that effort that increases the likelihood of a higher common outcome increases $i$'s payoff conditional on no-disclosure. Thus, deliberation protocols requiring some degree of consensus for disclosure are dominated by the unilateral disclosure protocol.

\begin{figure}

\begin{subfigure}{.475\linewidth}
  \includegraphics[width=\linewidth]{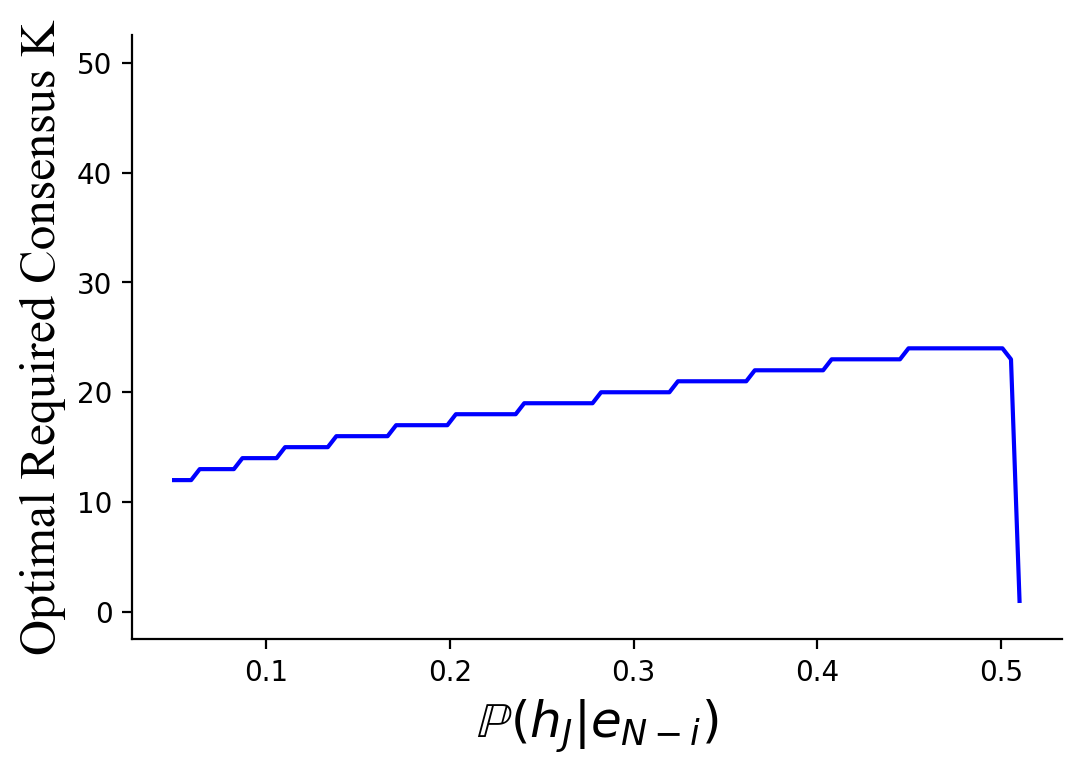}
  \caption{}
\end{subfigure}\hfill 
\begin{subfigure}{.475\linewidth}
  \includegraphics[width=\linewidth]{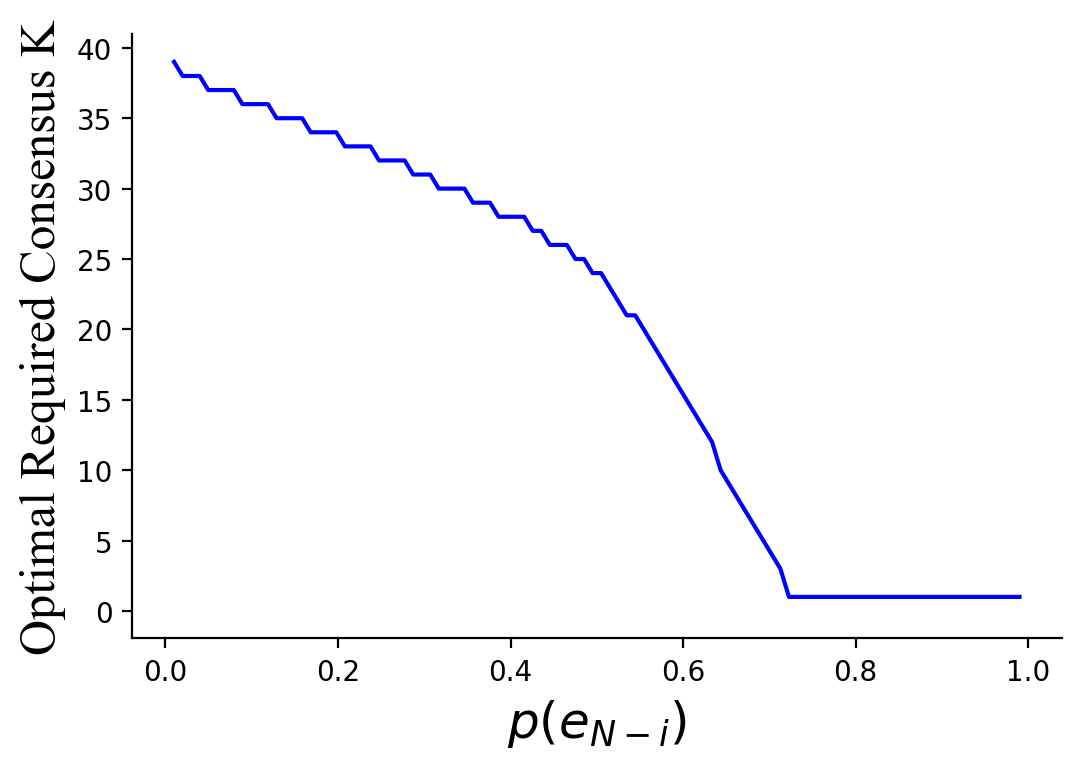}
  \caption{}
\end{subfigure}

\medskip 
\begin{subfigure}{.475\linewidth}
  \includegraphics[width=\linewidth]{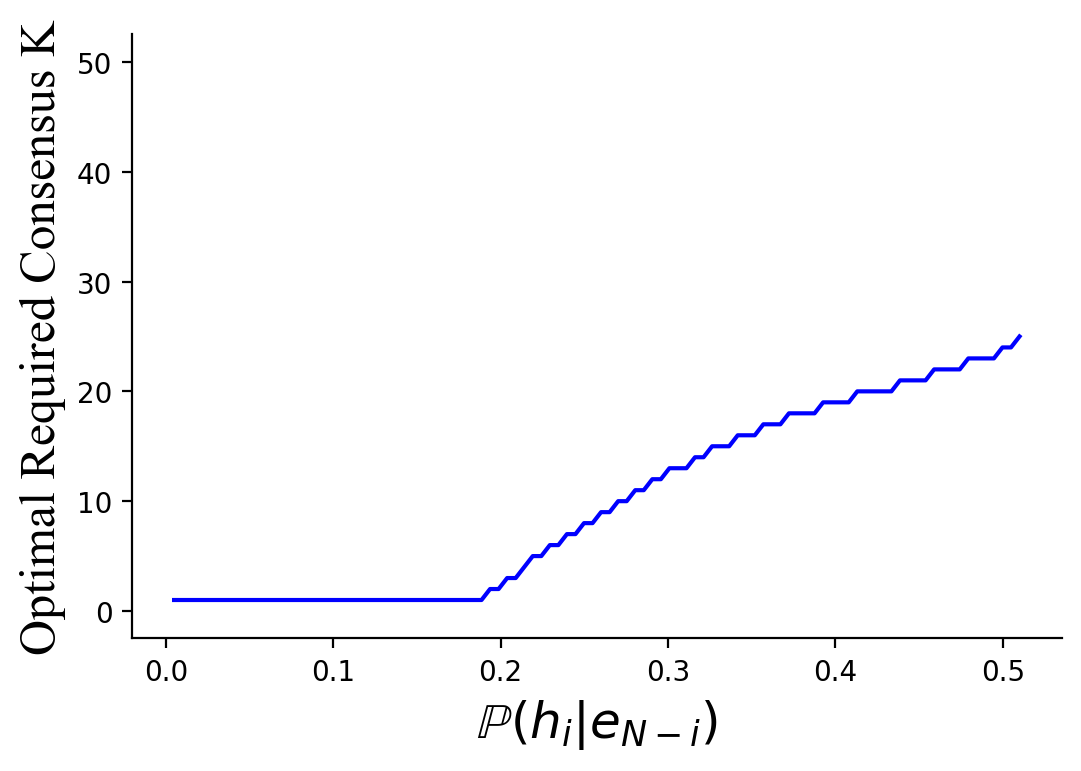}
  \caption{}
\end{subfigure}\hfill
\begin{subfigure}{.475\linewidth}
  \includegraphics[width=\linewidth]{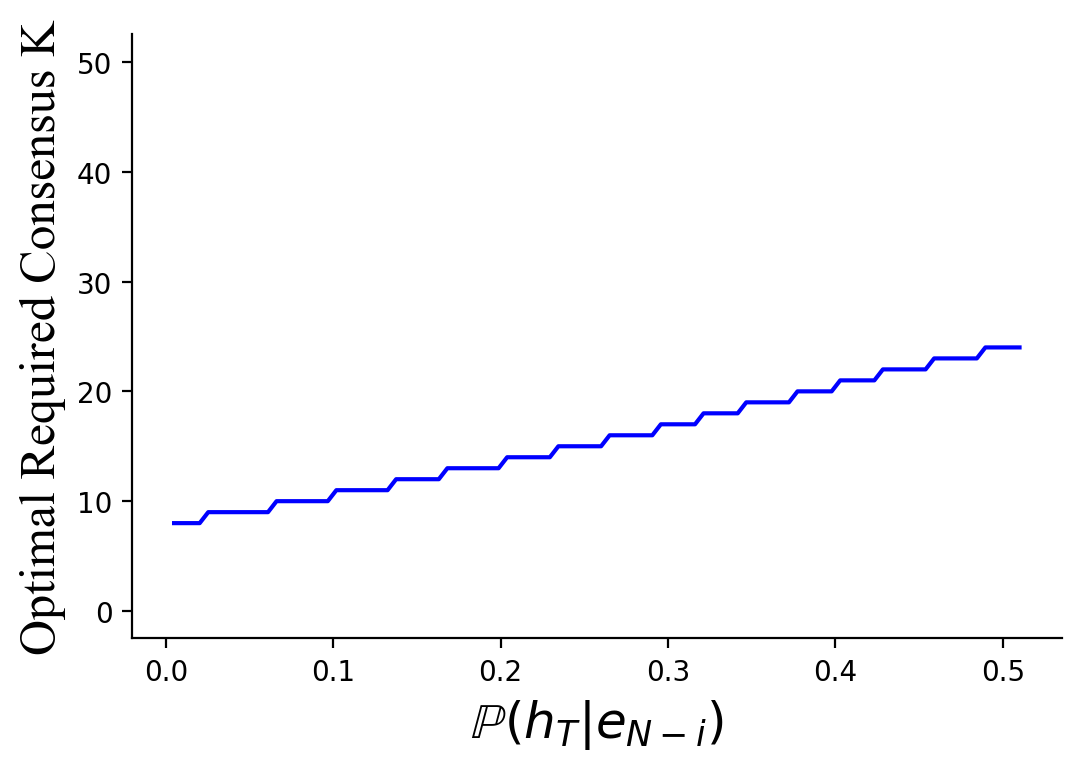}
  \caption{}
\end{subfigure}

\caption{In each panel, we relate the optimal need for consensus, i.e. the $K\leqslant N$ such that the $K$-majority rule maximizes the team's effort incentives. The baseline parameters are $p(e_{N\setminus i})=p(e_N)=.5$, $\mathbb{P}\left(h_i|e_N\right)=\mathbb{P}\left(h_j|e_N\right)=\mathbb{P}\left(h_T|e_N\right)=.51$, and $\mathbb{P}\left(h_i|e_{N\setminus i}\right)=\mathbb{P}\left(h_j|e_{N\setminus i}\right)=\mathbb{P}\left(h_T|e_{N\setminus i}\right)=.5$. In each panel, we maintain the baseline parameters and plot the optimal $K$ against varying full-effort parameter specifications.}
\label{fig:2}
\end{figure}

\subsubsection{Optimal Symmetric Deliberation} Finally, in a numerical exercise, we characterize the \emph{optimal symmetric deliberation protocol}, in the sense that it maximizes   maximizes the equilibrium full-effort set.
Our numerical results assess how the optimal threshold varies with the effect of team-members' efforts on their own probability of a good outcome, the probability of a good common outcome, the probability of others' good outcomes, and the relative probability of common versus independent outcomes across team-members. The results are stated below and presented in Figure \ref{fig:2}. We consider a fixed set of baseline parameters --- the outcome distribution conditional on full effort --- and vary each of the four relevant parameters in the outcome distribution conditional on all but one team-member exerting effort. Our numerical results, as stated below, are robust to all  parameter specifications we considered. Panels (a)-(d) in Figure \ref{fig:2} illustrate the relations established in the results (i)-(iv) stated below, in that order.

\begin{numres}
An optimal symmetric deliberation protocol exists, that maximizes the equilibrium full-effort set. 
Let $K^*\leqslant N$ be the required consensus in the optimal symmetric deliberation protocol. Then:
\begin{enumerate}
    \item[(i)]     The relation between $K^*$ and the effect of agent $i$'s effort on agent $j$'s high-outcome probability (for $i\neq j$) is U shaped. That is, fixing $\mathbb{P}(h_j|e_N)$, $K^*$ first increases then decreases as $\mathbb{P}(h_j|e_{N\setminus i})$ increases. 
    
    \item[(ii)]     $K^*$ is increasing in the effect of agent $i$'s effort on the probability of a common outcome. That is, fixing $p(e_N)$, $K^*$ is decreasing in $p(e_{N\setminus i})$. 

    \item[(iii)]     $K^*$ is decreasing in the effect of agent $i$'s effort on their own high-outcome probability. That is, fixing $\mathbb{P}(h_i|e_N)$, $K^*$ is decreasing in $\mathbb{P}(h_i|e_{N\setminus i})$. 

    \item[(iv)]     $K^*$ is decreasing in the effect of agent $i$'s effort on the probability of a common high-outcome. That is, fixing $\mathbb{P}(h_T|e_N)$, $K^*$ is decreasing in $\mathbb{P}(h_T|e_{N\setminus i})$. 
\end{enumerate}
\end{numres}



\section{Conclusion}
This paper studies production in teams where team-members are motivated by career concerns. In this context, the outcomes produced by a team have potentially distinct reputational implications for each team-member. These reputational implications depend not only on the exogenously given structure of team-production, but also on the deliberation protocol through which the team makes decisions --- specifically decisions to disclose team outcomes to third parties. 

The relevant contexts for team production are varied: much academic work is conducted in partnerships; and other creative fields, in which workers are typically motivated by the reputational implications of their production, are also mostly organized around teams. In all these contexts, team decisions are made according to a variety of decision protocols. For example, some teams have pre-assigned leaders --- think, for example, of academic work in which there is a clear ``first author.'' Other team decisions are made consensually, or at least through some degree of consensus between more or less ``equal'' parties. Our first contribution, in section \ref{sec:disc}, is to show that the protocol through which disclosure decisions are made affects the attribution of credit and blame for the team's successes and failures across the team members. We show that these properties of credit attribution imply team-disclosure equilibria in which not all outcomes produced by a team are disclosed to third-party observers; and we contrast this equilibrium structure with full-disclosure equilibria which are achieved in parallel disclosure problems when decisions are made by single individuals. 

Next, in section \ref{sec:inc}, we study the incentive properties of different \emph{deliberation protocols}, the organizational structure through which the team decides to disclose or not disclose their produced outcome. We characterize production environments in which partial-disclosure equilibria, induced by deliberation protocols where decisions are made through some degree of consensus, provide incentives for individuals to exert effort above and beyond their willingness to contribute to the team when team-outcomes are always disclosed. This result illustrates that, by requiring consensus in disclosure decisions, individual team members become more willing to exert effort that will not directly benefit themselves, but rather indirectly increase their payoffs by better aligning their interests with those of other team-members. 

The broader contributions of this paper are twofold. Firstly, we add to a large literature in organizational economics that studies how hierarchies and ownership rights shape productive incentives within teams. To the best of our knowledge, this is the first paper to observe that hierarchies and team-organization affects team-members' incentives not only directly through setting the teams goals and priorities; but also indirectly via determining how credit for the outcomes of the team is allocated across its various contributors. Our study of the incentive properties of different deliberation protocols highlights that, even in a world where contracts cannot be ``complete,''  team-performance can be improved through the design of decision rights. 

Secondly, we contribute to a nascent literature on mechanism design in teams. An alternative way to interpret the deliberation protocol is as a set of rules to which the team pre-commits, which then specifies which team-members' ``incentive constraints'' must be respected by the team. We view our exercise in designing the optimal deliberation protocol as a first step towards further work on mechanism design in teams under different team-decision protocols.

\section*{References}

\vskip5pt\noindent Austen-Smith, David, and Timothy Feddersen. (2005) ``Deliberation and Voting Rules.'' in: David Austen-Smith, and John Duggan (eds.), \emph{Social Choice and Strategic Decisions}, Springer: 269-316.

\vskip5pt\noindent Bardhi, Arjada, and Nina Bobkova. (2023) ``Local Evidence and Diversity in Minipublics.'' \emph{Journal of Political Economy}, forthcoming.

\vskip5pt\noindent Battaglini, Marco. (2022) ``Multiple Referrals and Multidimensional Cheap Talk," \emph{Econometrica} \textbf{70}: 1379-1401.

\vskip5pt\noindent Battaglini, Marco, and Thomas R. Palfrey. (2023) ``Organizing for Collective Action: Olson Revisited,'' \emph{working paper}.

\vskip5pt\noindent Baumann, Leonie, and Rohan Dutta. (2022) ``Strategic Evidence Disclosure in Networks and Equilibrium Discrimination,'' \emph{working paper}.

\vskip5pt\noindent Ben-Porath, Elchanan, Eddie Dekel, and Barton L. Lipman. (2018) ``Disclosure and Choice.'' \emph{Review of Economic Studies} \textbf{85}: 1471-1501.

\vskip5pt\noindent Bernheim, B. Douglas, Bezalel Peleg, and Michael D. Whinston. (1987) "Coalition-Proof Nash Equilibria i. Concepts." \emph{Journal of Economic Theory} \textbf{42}: 1-12.

\vskip5pt\noindent Crawford, Vincent P., and Joel Sobel. (1982) ``Strategic Information Transmission,'' \emph{Econometrica}: 1431-1451.

\vskip5pt\noindent Dye, Ronald A. (1985) ``Disclosure of Nonproprietary Information.'' \emph{Journal of Accounting Research}: 123-145.

\vskip5pt\noindent Dziuda, Wioletta. (2011) "Strategic Argumentation," \emph{Journal of Economic Theory} \textbf{146}: 1362-1397.

\vskip5pt\noindent Fortunato, S., Bergstrom, C.T., Börner, K., Evans, J.A., Helbing, D., Milojević, S., Petersen, A.M., Radicchi, F., Sinatra, R., Uzzi, B. and Vespignani, A. (2018). ``Science of Science,'' \emph{Science}, \textbf{359}: eaao0185.

\vskip5pt\noindent Gentzkow, Matthew, and Emir Kamenica. (2016) ``Competition in Persuasion,'' \emph{Review of Economic Studies} \textbf{84}: 300-322.

\vskip5pt\noindent Gerardi, Dino, and Leeat Yariv. (2007) ``Deliberative Voting.'' \emph{Journal of Economic Theory} \textbf{134}: 317-338.

\vskip5pt\noindent Gerardi, Dino, and Leeat Yariv. "Information Acquisition in Committees.'' \emph{Games and Economic Behavior} \textbf{62}: 436-459.

\vskip5pt\noindent Grossman, Sanford J. (1981) ``The Informational Role of Warranties and Private Disclosure about Product Quality.'' \emph{Journal of Law and Economics} \textbf{24}: 461-483.

\vskip5pt\noindent Hara, Kazuhiro. (2022) ``Coalitional Strategic Games.'' \emph{Journal of Economic Theory} \textbf{204}: 105512.

\vskip5pt\noindent Haghpanah, Nima, Aditya Kuvalekar, and Elliot Lipnowski. (2022) ``Selling to a Group.'' \emph{working paper}.

\vskip5pt\noindent Halac, Marina, Elliot Lipnowski, and Daniel Rappoport. (2021) ``Rank Uncertainty in Organizations.'' \emph{American Economic Review} \textbf{111}: 757-86.

\vskip5pt\noindent Hu, Peicong, and Joel Sobel. (2019) "Simultaneous versus Sequential Disclosure," \emph{working paper}.

\vskip5pt\noindent Jones, Benjamin. (2021) ``The Rise of Research Teams: Benefits and Costs in Economics.'' \emph{Journal of Economic Perspectives} \textbf{35}: 191-216.

\vskip5pt\noindent Kamenica, Emir, and Matthew Gentzkow. (2011) ``Bayesian Persuasion,'' \emph{American Economic Review} \textbf{101}:: 2590-2615.

\vskip5pt\noindent Kim, Jeongbin, Thomas R. Palfrey, and Jeffrey R. Zeidel. (2022) "Games Played by Teams of Players." \emph{American Economic Journal: Microeconomics} \textbf{14}: 122-57.

\vskip5pt\noindent David M. Kreps and Robert Wilson. (1982) ``Sequential Equilibria,'' \emph{Econometrica} \textbf{50}:863-894.

\vskip5pt\noindent Lazear, Edward P., and Kathryn L. Shaw. (2007) ``Personnel Economics: the Economist's View of Human Resources.'' \emph{Journal of Economic Perspectives} \textbf{21}: 91-114.

\vskip5pt\noindent Levy, Gilat. (2007) ``Decision Making in Committees: Transparency, Reputation, and
Voting Rules.'' \emph{American Economic Review}, \textbf{97}: 150-168.

\vskip5pt\noindent Martini, Giorgio. (2018) "Multidimensional Disclosure,'' \emph{working paper}.

\vskip5pt\noindent Matthews, Steven, and Andrew Postlewaite. (1985) ``Quality Testing and Disclosure,'' \emph{RAND Journal of Economics}: 328-340.

\vskip5pt\noindent Milgrom, Paul (1981) ``Good News and Bad News: Representation Theorems and Applications.'' \emph{Bell Journal of Economics}: 380-391.

\vskip5pt\noindent Milgrom, Paul. (2008) "What the Seller Won't Tell You: Persuasion and Disclosure in Markets." \emph{Journal of Economic Perspectives} \textbf{22}: 115-131.

\vskip5pt\noindent Onuchic, Paula. (2022) ``Advisors with Hidden Motives,'' \emph{working paper}. 

\vskip5pt\noindent Onuchic, Paula \textcircled{r} Debraj Ray. (2023) "Signaling and Discrimination in Collaborative Projects." \emph{American Economic Review} \textbf{113}: 210-252.

\vskip5pt\noindent Ray, Debraj \textcircled{r} Arthur Robson. (2018) ``Certified Random: A New Order for Coauthorship,''  \emph{American Economic Review} \textbf{108}: 489-520.

\vskip5pt\noindent Schwert, G. William. (2021) ``The Remarkable Growth in Financial Economics, 1974–2020,'' \emph{Journal of Financial Economics} \textbf{140}: 1008-1046.

\vskip5pt\noindent Shaked, Moshe, and J. George Shanthikumar, eds. \emph{Stochastic Orders}. New York, NY: Springer New York.

\vskip5pt\noindent Shishkin, Denis. (2021) ``Evidence Acquisition and Voluntary Disclosure,'' \emph{working paper}.

\vskip5pt\noindent Dai, Tianjiao, and Juuso Toikka. (2022) ``Robust Incentives for Teams.'' \emph{Econometrica} \textbf{90}: 1583-1613.

\vskip5pt\noindent Visser, Bauke, and Otto H. Swank. (2007) ``On Committees of Experts,'' \emph{Quarterly Journal of Economics} \textbf{122}: 337-372.

\vskip5pt\noindent Whitmeyer, Mark, and Kun Zhang. (2022) ``Costly Evidence and Discretionary Disclosure,'' \emph{working paper}.

\vskip5pt\noindent Winter, Eyal. (2004) ``Incentives and Discrimination.'' \emph{American Economic Review} \textbf{94}: 764-773.
\appendix
\section{Proofs}
\subsection{Proof of Proposition \ref{pr:1a}}
Suppose $(x,\omega^{ND})$ is an equilibrium-tuple, and suppose there is some $\omega\in\Omega$, with $\omega_i>\omega_i^{ND}$, but $x_i(\omega)<1$.
Consider an equilibrium candidate $(x',\omega'^{ND})$ where $x'=x$ everywhere, except $x'_i=1$. First note that, because $(x,\omega^{ND})$ is an equilibrium, it must be that $D(1,x_{-i}(\omega))=D(0,x_{-i}(\omega))$, and therefore the equilibrium team-disclosure strategy remains unchanged when we set $x'_i(\omega)=1$; and therefore $\omega'^{ND}=\omega^{ND}$. 

Next, note that because the original equilibrium is coalition-proof, the change from $x$ to $x'$ does not create incentives for any of the other team members $j\neq i$ to deviate from their original strategies $x_j(\omega)$. Specifically, coalition-proofness implies that either $D(1_{\{i,j\}},x(\omega)_{-\{i,j\}})=D(0_{\{i,j\}},x(\omega)_{-\{i,j\}})$ or $\omega_i<\omega_i^{ND}$; and, in either case, $x'_j(\omega)=x'_j(\omega)$ is a ``best response'' for $j$. Therefore, $(x',\omega'^{ND})$ is an equilibrium, equivalent to the original equilibrium. 

An analogous argument can be used to show that if $\omega_i<\omega_i^{ND}$, but $x_i(\omega)>1$ for some $\omega\in\Omega$ and $i\in N$, a modified equilibrium exists  where such $x_i(\omega)$ is set to $0$. By altering the original equilibrium successively for all such $\omega$'s and $i$'s, we reach a modified equilibrium in which every team member uses a threshold strategy, and which is equivalent to the original equilibrium.\qed

\subsection{Proof of Theorem \ref{th:1}}
We prove the three statements in the Theorem separately.
\subsubsection{Proof of Statement 1}
It is easy to see that a full-disclosure equilibrium always exists, where $x_i(\omega)=1$ for all $\omega\in\Omega$ and all $i\in N$, and $\omega_i^{ND}=\min\{\omega_i:\omega\in\Omega\}$. Given this vector of no-disclosure beliefs, ``always disclose'' is consistent with individual payoff maximization and coalition-proofness. The vector of no-disclosure beliefs is Bayes-consistent, because no-disclosure does not happen on-path. \qed 

\subsubsection{Proof of Statement 2}
First note that, if disclosure can be chosen unilaterally by all team members, then any candidate partial-disclosure equilibrium would unravel. Therefore, no partial-disclosure equilibrium exists in that case. 

To show the converse, we first define a map $\Phi:co(\Omega)\rightarrow 2^{co(\Omega)}$, as follows: For each $\bar{\omega}\in co(\Omega)$, $\hat{\omega}\in\Phi(\bar{\omega})$ if and only if there exists a vector of individual disclosure strategies $x$ satisfying $x_i(\omega)=0$ if $\omega_i=\min(\Omega_i)$, and
$$\omega_i>\bar{\omega}_i\Rightarrow x_i(\omega)=1\text{ and }\omega_i<\bar{\omega}_i\Rightarrow x_i(\omega)=0,$$
and such that 
$$\hat{\omega}_i=
\frac{\int_\Omega \omega_i(1-d(\omega))dF(\omega)}{\int_\Omega (1-d(\omega))dF(\omega)}\text{, where }d(\omega)=D(x(\omega)).
$$
In words, $\Phi$ maps each ``candidate vector'' of equilibrium no-disclosure posteriors into a vector of ``individually rational'' no-disclosure posteriors which is consistent with the starting candidate vector. These ``individually rational'' posteriors are those consistent with agents using threshold individual disclosure strategies, where they vote for disclosure if their own realized outcome is strictly better than the candidate no-disclosure value and vote against disclosure if the realized outcome is strictly worse than the candidate value. We allow individuals to use any mixed strategy if their realized outcome equals their candidate no-disclosure posterior, with the exception that individuals always vote to not disclose if their worst possible outcome realizes.

First note that $\Phi(\bar{\omega})$ is non-empty for every $\bar{\omega}\in co(\Omega)$, because no-disclosure happens on path for all the described strategies --- at the very least, all agents vote for non-disclosure when $\omega=(\min(\Omega_1),...,\min(\Omega_N))$, and the team chooses no disclosure by consensus.

Now observe that, because the construction of $\Phi$ allows individuals to use any mixed strategy when their realized outcome equals their candidate no-disclosure posterior, and because $D$ is consistent with mixed strategies (see part 3 in Assumption \ref{as:1}), then $\Phi(\bar{\omega})$ is a closed set for all $\bar{\omega}\in co(\Omega)$; and $\Phi$ is upper-hemicontinuous. 
Therefore, $\Phi$ has a closed graph, and by the Kakutani fixed point theorem, $\Phi$ has a fixed point in $co(\Omega)$. It is easy to see that a fixed point of $\Phi$ defines an equilibrium of the team-disclosure game. 

Now we argue that, if disclosure cannot be chosen unilaterally by all team members, then there must be a fixed point $w$ of $\Phi$ with $w_i>\min(\Omega_i)$ for some $i\in N$. To that end, let $w\in\Phi(w)$ be a fixed point of $\Phi$. Then it must be that there is a vector of individual disclosure strategies $x$ satisfying $x_i(\omega)=0$ if $\omega_i=\min(\Omega_i)$ such that for every $i\in N$,
$$w_i=
\frac{\int_\Omega \omega_i(1-d(\omega))dF(\omega)}{\int_\Omega (1-d(\omega))dF(\omega)}\text{, where }d(\omega)=D(x(\omega)).
$$
But let $i\in N$ be a team member who cannot unilaterally choose disclosure. Then it must be that all realizations $\omega$ with $\omega_j=\min(\Omega_j)$ for every $j\neq i$ are not disclosed --- regardless of the realization of $\omega_i$. Consequently, every possible realized outcome for individual $i$ is not disclosed with positive probability, and therefore $w_i>\min(\Omega_i)$. This fixed-point of $\Phi$ thus defines a partial-disclosure equilibrium.\qed

\subsubsection{Proof of Statement 3}
First suppose $D$ is such that disclosure can be chosen unilaterally by some team-member $i$. It is easy to see that any interior equilibrium, in which $x^{ND}_i>\min(\Omega_i)$ would unravel. Therefore, no interior equilibrium exists. To show the converse result, assume now that $D$ is such that disclosure cannot be chosen unilaterally by any team member. 

Suppose the vectors of individual disclosure strategies $x$ and no-disclosure posteriors $\omega^{ND}$ constitute a partial-disclosure equilibrium. 
By Proposition \ref{pr:1a}, we know that an equilibrium exists which is equivalent to $(x,\omega^{ND})$, call it $(x',\omega^{ND})$ in which every individual disclosure strategy depends only on the agents' own outcomes, and further $x'_i$ is a threshold disclosure strategy, for every $i\in N$. That is, for every $i\in N$, 
$$\omega_i>\omega^{ND}_i\Rightarrow x'_i(\omega)=1\text{ and }\omega_i<\omega^{ND}_i\Rightarrow x'_i(\omega)=0.$$

Now let $I\subseteq N$ be the set of agents $i$ such that $\omega_i^{ND}>\min(\Omega_i)$. If $I=N$, then this equilibrium is an interior equilibrium, and we are done. Suppose instead that $I\subset N$. Because $(x',\omega^{ND})$ is a partial-disclosure equilibrium, and because $x'_i$ is a threshold strategy for every $i\in N$, it must be that $d(\ushort{\omega})<1$, where $\ushort{\omega}=(\min(\Omega_1),...,\min(\Omega_N))$. Then let $O\subseteq N$ be the set of agents $i$ such that $x'_i(\ushort{\omega})<1)$.

First note that $I\subseteq O$, for $x'_i$ is a threshold strategy for every $i'\in N$. Now consider two cases: (i) $O=N$, and (ii) $O\subset N$. In case (i), then because disclosure cannot be chosen unilaterally, and because each individual's disclosure strategy $x'_i$ depends only on their own outcome, it must be that for every $i\in N$
$$D(\omega_i,\ushort{\omega}_{-i})<1\text{, for every }\omega_i\in\Omega_i.$$
But this implies that for every $i\in N$, $\omega^{ND}_i>\min(\Omega_i)$, which contradicts the assumption that $I\neq N$. Now take case (ii), where $O\subset N$. Then, because $I\subseteq O$, there must be some $j\in(N\setminus O)\cap(N\setminus I)$. Then it must be that 
$$D(\omega_j,\ushort{\omega}_{-j})=D(\ushort{\omega})<1,$$
and therefore $\omega^{ND}_j>\min(\Omega_j)$, which contradicts the assumption that $j\in N\setminus I$. 

And so it must be that $I=N$, and therefore $(x',\omega^{ND})$ --- and consequently the equivalent equilibrium $(x,\omega^{ND})$ --- are interior.

\qed

\subsection{Proof of Theorem \ref{th:2}}
The proof combines the following two Lemmas. 

\begin{lemma} If disclosing does not require more consensus than concelaing, then there exists a full disclosure equilibrium that is consistent with deliberation.
\label{lem:a}
\end{lemma}

\begin{proof}[Proof of Lemma \ref{lem:a}] Suppose $D$ is such that disclosing \emph{does not} require more consensus than concealing. Then there exists some subgroup $I\subset N$ such that $D(1_{I},0_{-I})=1$, $D(0_{I},1_{-I})=0$ and, for all $J\subset I$, $D(0_{J},1_{-J})=1$. Now consider a candidate full-disclosure equilibrium where $x_i(\omega)=1$ for all $\omega\in\Omega$ and all $i\in I$ (we will specify the other agents' individual disclosure strategies later). Given subgroup $I$'s disclosure strategies, and the deliberation protocol, this candidate equilibrium indeed has full-disclosure. Moreover, in this candidate equilibrium, we conjecture that the off-path no-disclosure posteriors are $\omega^{ND}_i=\min(\Omega_i)$ for each $i\in I$.

Now we want to build individual disclosure strategies to be used to ``justify'' the off-path equilibrium posteriors. To that end, 
construct individual disclosure strategies such that for every $i\in I$, 
$$\hat{x}_i(\omega)=0\text{ if and only if }\omega_i=\min(\Omega_i).$$
And for every $j\in N\setminus I$, 
$$\hat{x}_j(\omega)=1\text{ for all }\omega\in\Omega.$$
It is easy to see that, given the vector of disclosure rules $\hat{x}$ and the deliberation protocol $D$, the team disclosure satisfies $d(\omega)=D(\hat{x}(\omega))=0$ if and only if $\omega_i=\min(\Omega_i)$ for all $i\in I$. And therefore Bayesian plausibility implies that $\omega^{ND}_i=\min(\Omega_i)$ for all $i\in I$.
To complete the construction of the equilibrium, for every $j\in N\setminus I$, let $\omega^{ND}_j$ be the Bayes-consistent no-disclosure posterior implied by $\hat{x}$. And let their equilibrium individual disclosure strategy be given by
$$x_j(\omega)=1\Leftrightarrow \omega_j\geqslant \omega^{ND}_j\text{, for every }j\in N\setminus I.$$
\end{proof}

\begin{lemma} If disclosure requires more consensus than concealing, then there is no full-disclosure equilibrium that is consistent with deliberation.
\label{lem:b}
\end{lemma}
\begin{proof}[Proof of Lemma \ref{lem:b}] Let $D$ be such that disclosure requires more consensus than concealing. And suppose a vector $x$ of individual disclosure strategies and a vector $\omega^{ND}$ of no-disclosure posteriors constitute a full-disclosure equilibrium. Let $I\subset N$ be the largest subgroup of team members such that 
$$\omega^{ND}_i=\min(\Omega_i)\text{ for all }i\in I.$$
\begin{claim} The set $I$ is non-empty, and  $D(1_{I},0_{-I})=1$.
\end{claim}
\begin{proof}[Proof of Claim]
\label{cl:1}
Suppose towards a contradiction that $D(1_{I},0_{-I})=0$ (which would vacuously hold if $I$ were empty). Then every member of subgroup $-I=N\setminus I$ strictly prefers to not disclose the realization $\ushort{\omega}$, where $\ushort{\omega}_i=\min(\Omega_i)$ for every $i\in N$. And moreover, because $D(1_{I},0_{-I})=0$, the subgroup $-I$ is able to block the disclosure of $\ushort{\omega}$. Therefore, by coalition-proofness, it must be that $\ushort{\omega}$ is not disclosed in equilibrium. But this contradicts the assumption that the starting equilibrium has full-disclosure.
\end{proof}

Now take a vector of individual disclosure rules $\hat{x}$, with each individual strategy depending only on their own realized outcome, to be used as a candidate to ``justify'' the off-path no-disclosure beliefs $\omega^{ND}$. Take some $\hat{\omega}\in\Omega$ with $\hat{\omega}_i=\min(\Omega_i)$ for every $i\in I$, and such that $d(\hat{\omega})=D(\hat{x}(\hat{\omega}))<1$ --- such a $\hat{\omega}$ must exist if $\hat{x}$ is to justify the conjectured no-disclosure beliefs. Let $I'$ be the set of agents $i\in N$ such that $\hat{x}_i(\hat{\omega})<1$. We now consider two cases.

\vskip15pt
\noindent\emph{Case 1.} Suppose there is some $i^*\in I\setminus I'$; that is, there is some $i^*\in I$ such that $\hat{x}_{i^*}(\hat{\omega})=1$. Then there must exist some $\hat{\omega}'$ with $\hat{\omega}_i=\hat{\omega}'_i$ for all $i\in N\setminus\{i^*\}$ such that $d(\hat{\omega}')=D(\hat{x}(\hat{\omega}')<1$ (because each individual strategy depends only on their own realized outcome). But note that, because $\hat{\omega}'_{i^*}\neq \hat{\omega}_{i^*}$, then it must be that $\hat{\omega}'_{i^*}>\min(\Omega_{i^*})$; and therefore the no-disclosure posterior about $i^*$'s outcome implied by $\hat{x}$ cannot be $\min(\Omega_{i^*})$.

\vskip15pt

\noindent\emph{Case 2.} Suppose instead that $I\subseteq I'$ (and therefore $I\setminus I'=\varnothing$). In this case, we will use the following claim.
\begin{claim}
\label{cl:2}
If $I\subseteq I'$, then there exists some $J\subset I'$ with $I\setminus J\neq \varnothing$ such that $D(0_{J},1_{-J})=0$. 
\end{claim}
\begin{proof}[Proof of Claim]
To see this, first note that $I'$ is a pivotal subgroup: $D(1_{I'},0_{-I'})=1$ because $D(1_{I},0_{-I})=1$ by Claim \ref{cl:1}, and $I\subseteq I'$; and $D(0_{I'},1_{-I'})=0$ by construction, since we assumed that $d(\hat{\omega})<1$. If $I$ is itself a pivotal subgroup, then some $J\subset I$ exists such that $D(0_{J},1_{-J})=0$ --- because $D$ is such that disclosing requires more consensus than concealing --- confirming the statement in the claim. 

Suppose instead that $I$ is not a pivotal subgroup. And take some $J\subset I'$ such that $D(0_{J},1_{-J})=0$ (such a $J$ exists because $I'$ is pivotal). We know that $J\neq I$, for otherwise $I$ would be a pivotal subgroup. Then there are two remaining possibilities: if $I\setminus J\neq \varnothing$, then the statement in the claim is confirmed. 

If otherwise $I\subset J$, then $J$ itself is a pivotal subgroup --- $D(0_{J},1_{-J})=0$ by assumption, and $D(1_{J},0_{-J})=1$ because $I\subset J$. And therefore we can take a subgroup $J'\subset J$ with $D(0_{J'},1_{-J'})=0$. If $I\subset J'\neq \varnothing$, we are done. If not, the procedure can be repeated until such a subset is found. We know that this recursive approach will eventually yield such a subset because the total size of the group is finite.
\end{proof}
\vskip15pt
Now take some $i^*\in I\setminus J$. Then there must exist some $\hat{\omega}'$ with $\hat{\omega}_i=\hat{\omega}'_i$ for all $i\in N\setminus\{i^*\}$ such that $d(\hat{\omega}')=D(\hat{x}(\hat{\omega}')<1$ (because each individual strategy depends only on their own realized outcome). But note that, because $\hat{\omega}'_{i^*}\neq \hat{\omega}_{i^*}$, then it must be that $\hat{\omega}'_{i^*}>\min(\Omega_{i^*})$; and therefore the no-disclosure posterior about $i^*$'s outcome implied by $\hat{x}$ cannot be $\min(\Omega_{i^*})$.
\vskip15pt

Combining Case 1 and Case 2, we conclude that there is no vector of individual disclosure rules $\hat{x}$, with each individual strategy depending only on their own realized outcome, that can ``justify'' the conjectured no-disclosure posteriors as consistent with the deliberation protocol. And this is true for any conjectured full-disclosure equilibrium. Consequently, there is no full-disclosure equilibrium that is consistent with deliberation. \end{proof}

\subsection{Proof of Proposition \ref{pr:1}}
By Theorem \ref{th:1}, we know that an interior equilibrium exists whenever $D$ is such that disclosure cannot be chosen unilaterally. Now take some interior equilibrium, with individual disclosure strategies $x$ and no-disclosure posteriors $\omega^{ND}$, where $\min_(\Omega_i)<\omega^{ND}_i<\max(\Omega_i)$ for every $i\in N$. Then individual payoff maximization and coalition-proofness implies that
\begin{enumerate}
\item If $D(1_{H(\omega)},0_{L(\omega)})=1$, then 
$$d(\omega)=D(x(\omega))=1,$$
for otherwise every member $i$ of $H(\omega)$ would strictly benefit from deviating to $x_i(\omega)=1$.
\item If $D(1_{H(\omega)},0_{L(\omega)})=0$,
$$d(\omega)=D(x(\omega))=0,$$
for otherwise every member $i$ of $L(\omega)$ would strictly benefit from deviating to $x_i(\omega)=0$.
\end{enumerate}
\qed

\subsection{Proof of Theorem \ref{th:3}}
Fix a vector of effort costs $c\in\mathbb{R}_{++}^N$. Suppose team member $i$ anticipates that every other team member will choose $e_j=1$ (for $j\neq i$), and that the team disclosure strategy will be $d$. Then $i$'s payoff from choosing effort $e_i=1$ is
$$\int_\Omega\omega_i d(\omega)dF(\omega;e_N)+\int_{\Omega}(1-d(\omega))\omega_i^{ND}dF(\omega;e_N)-c_i$$
$$=\int_\Omega\omega_i d(\omega)dF(\omega;e_N)+\int_{\Omega}(1-d(\omega))\left[\frac{\int_\Omega\omega_i(1-d(\omega))dF(\omega;e^1)}{\int_\Omega(1-d(\omega))dF(\omega;e_N)}\right]dF(\omega;e_N)-c_i$$
$$=\int_\Omega\omega_idF(\omega;e_N)+\int_{\Omega}(1-d(\omega))\left[\frac{\int_\Omega\omega_i(1-d(\omega))dF(\omega;e_N)}{\int_\Omega(1-d(\omega))dF(\omega;e_N)}-\omega_i\right]dF(\omega;e_N)-c_i$$
$$=\int_\Omega\omega_idF(\omega;e_N)-c_i,$$
where the last equality uses the fact that $\int_{\Omega}(1-d(\omega))\left[\frac{\int_\Omega\omega_i(1-d(\omega))dF(\omega;e_N)}{\int_\Omega(1-d(\omega))dF(\omega;e_N)}-\omega_i\right]dF(\omega; e_N)=0$.

And their payoff from choosing effort $e_i=0$ is
\begin{align*}
&\int_\Omega\omega_idF(\omega;e_{N\setminus i})+\int_{\Omega}(1-d(\omega))\left[\frac{\int_\Omega\omega_i(1-d(\omega))dF(\omega;e_N)}{\int_\Omega(1-d(\omega))dF(\omega;e_N)}-\omega_i\right]dF(\omega;e_{N\setminus i}),
\end{align*}
where note that in the second term the distribution of outcomes is affected by $i$'s effort choice, but the value of $\omega^{ND}_i$ is still calculated under the presumption that $e_i=1$, for the deviation to $e_i=0$ is not seen by the observer.

Therefore, given team-disclosure strategy $d$, there is an equilibrium of the effort-choice stage where every team member exerts effort if and only if for every $i\in N$,
$$\int_\Omega\omega_i \left[dF(\omega; e_N)-dF(\omega;e_{N\setminus i}))\right]+\int_{\Omega}(1-d(\omega))\left[\omega_i-\frac{\int_\Omega\omega_i(1-d(\omega))dF(\omega;e_N)}{\int_\Omega(1-d(\omega))dF(\omega;e_N)}\right]dF(\omega;e_{N\setminus i})$$
$$\geqslant c_i.$$
Or equivalently if and only if
$$-\int_\Omega(1-d(\omega))dF(\omega;e_{N\setminus i})\left[\frac{\int_\Omega\omega_i(1-d(\omega))dF(\omega;e_N)}{\int_\Omega(1-d(\omega))dF(\omega;e_N)}-\frac{\int_\Omega\omega_i(1-d(\omega))dF(\omega;e_{N\setminus i})}{\int_\Omega(1-d(\omega))dF(\omega;e_{N\setminus i})}\right]$$
\begin{equation}
\nonumber
+\int_\Omega\omega_i \left[dF(\omega; e_N)-dF(\omega;e_{N\setminus i}))\right]
\geqslant c_i\text{, for every }i\in N.
\end{equation}

\qed

\subsection{Rewriting Equation (\ref{eq:p1}) as (\ref{eq:q1})}
The left-hand side of equation (\ref{eq:p1}) is
$$\mathbb{E}\left[\omega_i|e_N\right]-\mathbb{E}\left[\omega_i|e_{N\setminus i}\right]-\mathbb{P}\left[ND|e_{N\setminus i}\right]\bigg\{\mathbb{E}\left[\omega_i|ND; e_N\right]-\mathbb{E}\left[\omega_i|ND; e_{N\setminus i}\right]\bigg\}.$$
Or equivalently,
$$\int_\Omega\omega_idF(\omega;e_N)-\int_\Omega\omega_idF(\omega;e_{N\setminus i})$$
$$+\int_\Omega\omega_i(1-d(\omega))dF(\omega;e_{N\setminus i})-\frac{\int_\Omega(1-d(\omega))dF(\omega;e_{N\setminus i})}{\int_\Omega(1-d(\omega))dF(\omega;e_N)}\int_\Omega\omega_i(1-d(\omega))dF(\omega;e_N)$$
\vskip10pt

$$=\int_\Omega d(\omega)dF(\omega;e_{N\setminus i})\left[\int_\Omega\omega_idF(\omega;e_N)-\int_\Omega\omega_idF(\omega;e_{N\setminus i})\right]$$
$$+\int_\Omega (1-d(\omega))dF(\omega;e_{N\setminus i})\int_\Omega\omega_idF(\omega;e_N)-\int_\Omega (1-d(\omega))dF(\omega;e_{N\setminus i})\int_\Omega\omega_idF(\omega;e_{N\setminus i})$$
$$+\int_\Omega\omega_i(1-d(\omega))dF(\omega;e_{N\setminus i})-\frac{\int_\Omega(1-d(\omega))dF(\omega;e_{N\setminus i})}{\int_\Omega(1-d(\omega))dF(\omega;e_N)}\int_\Omega\omega_i(1-d(\omega))dF(\omega;e_N)$$
\vskip10pt

$$=\int_\Omega d(\omega)dF(\omega;e_{N\setminus i})\left[\int_\Omega\omega_idF(\omega;e_N)-\int_\Omega\omega_idF(\omega;e_{N\setminus i})\right]$$
$$+\int_\Omega\omega_i(1-d(\omega))dF(\omega;e_{N\setminus i})-\int_\Omega (1-d(\omega))dF(\omega;e_{N\setminus i})\int_\Omega\omega_idF(\omega;e_{N\setminus i})$$
$$+\frac{\int_\Omega(1-d(\omega))dF(\omega;e_{N\setminus i})}{\int_\Omega(1-d(\omega))dF(\omega;e_N)}\bigg[\int_\Omega (1-d(\omega))dF(\omega;e_{N})\int_\Omega\omega_idF(\omega;e_N)$$
$$-\int_\Omega\omega_i(1-d(\omega))dF(\omega;e_N)\bigg]$$

$$=\big(1-\mathbb{P}\left[ND|e_{N\setminus i}\right]\big)\big(\mathbb{E}\left[\omega_i|e_N\right]-\mathbb{E}\left[\omega_i|e_{N\setminus i}\right]\big)+Cov\left[\omega_i,(1-d)|e_{N\setminus i}\right]$$
$$-\frac{\mathbb{P}\left[ND|e_{N\setminus i}\right]}{\mathbb{P}\left[ND|e_{N}\right]}Cov\left[\omega_i,(1-d)|e_N\right]$$

\begin{align*}
&=\big(1-\mathbb{P}\left[ND|e_{N\setminus i}\right]\big)\big(\mathbb{E}\left[\omega_i|e_N\right]-\mathbb{E}\left[\omega_i|e_{N\setminus i}\right]\big)\\[1em]
&\hskip110pt+\frac{\mathbb{P}\left[ND|e_{N\setminus i}\right]}{\mathbb{P}\left[ND|e_{N}\right]}Cov\left[\omega_i,d|e_N\right]-Cov\left[\omega_i,d|e_{N\setminus i}\right],\end{align*}
which is the expression in (\ref{eq:q1}).
\subsection{Proof of Proposition \ref{pr:3}}
\subsubsection{Proof of Statement 1}
The proof of statement 1 uses the following lemma. 

\begin{lemma}
\label{lem:a1}
If effort is self-improving, then for any equilibrium disclosure rule $d$ with $d(\omega)<1$ for some $\omega\in\Omega$, we have that for all $i\in N$
$$\mathbb{E}\left[\omega_i|ND;e_N\right]>\mathbb{E}\left[\omega_i|ND;e_{N\setminus i}\right].$$
\end{lemma}
\begin{proof}[Proof of Lemma]
For notational convenience, we fix some team-member $i\in N$, and let $G_{\nu}=F_i(\cdot|\nu;e_N)$, $\hat{G}_{\nu}=F_i(\cdot|\nu;e_{N\setminus i})$ for all $\nu\in\Omega_{N\setminus i}$ and $H=F_{N\setminus i}(\cdot;e_N)=F_{N\setminus i}(\cdot;e_{N\setminus i})$. Therefore, for any equilibrium disclosure rule $d$ with $d$ with $d(\omega)<1$ for some $\omega\in\Omega$, we have 
\begin{equation}
\label{eq:l1}
\mathbb{E}\left[\omega_i|ND;e_{N\setminus i}\right]=\frac{\int_{\Omega_{N\setminus i}}\int_{\Omega_i}\omega_i (1-d(\omega_i,\nu)) d\hat{G}_{\nu}dH(\nu)}{\int_{\Omega_{N\setminus i}}\int_{\Omega_i}(1-d(\omega_i,\nu)) d\hat{G}_{\nu}dH(\nu)}.
\end{equation}
\begin{equation}
\label{eq:l2}
\mathbb{E}\left[\omega_i|ND;e_N\right]=\frac{\int_{\Omega_{N\setminus i}}\int_{\Omega_i}\omega_i (1-d(\omega_i,\nu)) dG_{\nu}dH(\nu)}{\int_{\Omega_{N\setminus i}}\int_{\Omega_i}(1-d(\omega_i,\nu)) dG_{\nu}dH(\nu)}.
\end{equation}
Now observe that for every such disclosure rule $d$, and every realization $\nu\in\Omega_{N\setminus i}$ of the outcomes of team-members $N\setminus i$, there is some probability $\alpha(\nu)$ that the outcome is \emph{not disclosed} regardless of the realization $\omega_i$, some probability $\beta(\nu)$ that the realization $\omega_i$ (and therefore $i$'s individual disclosure strategy) is pivotal for the team-disclosure decision, and some probability $1-\alpha(\nu)-\beta(\nu)$ that the outcome is disclosed, regardless of the realization $\omega_i$. Using this, and the fact that any equilibrium disclosure rule can be reached by all team-members using a threshold individual disclosure strategy (as in Proposition \ref{pr:1a}), we can rewrite (\ref{eq:l1}) as
\begin{equation}
\label{eq:l3}
\mathbb{E}\left[\omega_i|ND;e_{N\setminus i}\right]=\frac{\int_{\Omega_{N\setminus i}}\alpha(\nu)\int_{\Omega_i}\omega_i d\hat{G}_{\nu}+\beta(\nu)\int_{\omega_i\leqslant \omega_i^{ND}}\omega_i d\hat{G}_{\nu}dH(\nu)}{\int_{\Omega_{N\setminus i}}\alpha(\nu)\int_{\Omega_i}d\hat{G}_{\nu}+\beta(\nu)\int_{\omega_i\leqslant \omega_i^{ND}}d\hat{G}_{\nu}dH(\nu)},
\end{equation}
where remember that $\omega_i^{ND}=\mathbb{E}\left[\omega_i|ND;e_N\right]$ in any equilibrium of the team-disclosure game where full effort is implemented.
Of course, we can rewrite (\ref{eq:l2}) analogously. Now let $\hat{G}^{-1}_\nu$ be the quantile function implied by $\hat{G}_\nu$. We can rewrite (\ref{eq:l3}) as 
\begin{equation}
\label{eq:l4}
\mathbb{E}\left[\omega_i|ND;e_{N\setminus i}\right]=\frac{\int_{\Omega_{N\setminus i}}\alpha(\nu)\int_0^1\hat{D}_\nu^{-1}(q)dq+\beta(\nu)\int_{0}^{\hat{G}_\nu(\omega_i^{ND})}\hat{G}^{-1}_\nu(q)dqdH(\nu)}{\int_{\Omega_{N\setminus i}}\alpha(\nu)+\beta(\nu)\hat{G}_\nu(\omega_i^{ND})dH(\nu)}
\end{equation}
Now suppose by contradiction that $\mathbb{E}\left[\omega_i|ND;e_{N\setminus i}\right]\geqslant\mathbb{E}\left[\omega_i|ND;e_{N}\right]=\omega_i^{ND}$. Because effort is self-improving, we have $\hat{G}_\nu(\omega_i^{ND})>G_\nu(\omega_i^{ND})$ and therefore 
\begin{equation}
\label{eq:l5}
\mathbb{E}\left[\omega_i|ND;e_{N\setminus i}\right]<\frac{\int_{\Omega_{N\setminus i}}\alpha(\nu)\int_0^1\hat{G}_\nu^{-1}(q)dq+\beta(\nu)\int_{0}^{G_\nu(\omega_i^{ND})}\hat{G}^{-1}_\nu(q)dqdH(\nu)}{\int_{\Omega_{N\setminus i}}\alpha(\nu)+\beta(\nu)G_\nu(\omega_i^{ND})dH(\nu)},
\end{equation}
where we used the fact that $\mathbb{E}\left[\omega_i|ND;e_{N\setminus i}\right]\geqslant\omega_i^{ND}$; which implies that the right-hand side of (\ref{eq:l5}) can be reached by removing ``worse than average'' realizations of $\omega_i$ from the average computed in (\ref{eq:l4}). But also note that, because $\hat{G}_\nu\succ_{FOS}G_\nu$ for every $\nu\in\Omega_{N\setminus i}$, we have $G^{-1}_\nu(q)\geqslant\hat{G}^{-1}_\nu(q)$ for every $q\in[0,1]$. And consequently
\begin{align*}
&\frac{\int_{\Omega_{N\setminus i}}\alpha(\nu)\int_0^1\hat{G}_\nu^{-1}(q)dq+\beta(\nu)\int_{0}^{G_\nu(\omega_i^{ND})}\hat{G}^{-1}_\nu(q)dqdH(\nu)}{\int_{\Omega_{N\setminus i}}\alpha(\nu)+\beta(\nu)G_\nu(\omega_i^{ND})dH(\nu)}\leqslant\\[1em]&\hskip80pt\frac{\int_{\Omega_{N\setminus i}}\alpha(\nu)\int_0^1G_\nu^{-1}(q)dq+\beta(\nu)\int_{0}^{G_\nu(\omega_i^{ND})}G^{-1}_\nu(q)dqdH(\nu)}{\int_{\Omega_{N\setminus i}}\alpha(\nu)+\beta(\nu)G_\nu(\omega_i^{ND})dH(\nu)}=\mathbb{E}\left[\omega_i|ND;e_N\right]
\end{align*}
Now combining this with (\ref{eq:l5}), we have 
$$\mathbb{E}\left[\omega_i|ND;e_{N\setminus i}\right]<\mathbb{E}\left[\omega_i|ND;e_N\right],$$
which contradicts the assumption that $\mathbb{E}\left[\omega_i|ND;e_{N\setminus i}\right]\geqslant\mathbb{E}\left[\omega_i|ND;e_{N}\right]$; thereby proving the statement of the lemma.
\end{proof}

Using Lemma \ref{lem:a1} we therefore know that for any equilibrium disclosure rule $d$ with $d(\omega)<1$ for some $\omega\in\Omega$, the third term in the left-hand side of equation (\ref{eq:p1}) --- in Theorem \ref{th:3} --- is negative. And therefore, by Theorem \ref{th:3}, $fe(d)\subset fe(d')$, where $d'$ is the full-disclosure rule. And consequently $FE(D)\subset FE(D')$, where $D'$ is the unilateral disclosure protocol and $D$ is any other deliberation protocol.

\subsubsection{Proof of Statement 2}
First note that for any team-disclosure rule $d$ and any effort vector $e$, we can write
\begin{equation}
\label{eq:11}\mathbb{E}\left[\omega_i|ND;e\right]=\frac{\int_{\Omega_i}\omega_i(1-d_i(\omega_i))dF_i(\omega_i;e)}{\int_{\Omega_i}(1-d_i(\omega_i))dF_i(\omega_i;e)},
\end{equation}
where $F_i(\cdot;e)$ is the marginal distribution of $\omega_i$ given $e$, and $d_i(\omega_i)=\int_{\Omega}d(\omega)dF(\omega|\omega_i;e)$
is the overall probability that a $\omega_i$ would be disclosed (integrating over all the possible outcome realizations for other team members, given that $\omega_i$ happens).

Now consider the consensual-disclosure deliberation process. By Theorem \ref{th:1} and Proposition \ref{pr:1a}, there exists an interior equilibrium in which every team-member $i$ uses a threshold individual disclosure strategy; so that for every $i\in N$, $\omega_i>\omega^{ND}_i\Rightarrow x_i(\omega)=1\text{ and }\omega_i<\omega^{ND}_i\Rightarrow x_i(\omega)=0$, where $\omega_i^{ND}=\mathbb{E}\left[\omega_i|ND,e_N\right]$. For ease of exposition, assume that for each $i\in N$ the equilibrium $\omega_i^{ND}\not\in\Omega_i$, so that we can express $i$'s individual disclosure strategy without loss as $x_i(\omega)=1$ if $\omega_i>\omega^{ND}_i$ and $x_i(\omega)=0$ otherwise. The proof works analogously if $\omega^{ND}_i\in\Omega_i$ for any $i\in N$.

Because the disclosure must be consensually chosen by all team members, we have that for each $i\in N$, and a given effort vector $e$,
$$d_i(\omega_i;e)=\begin{cases}
0\text{, if }\omega_i\leqslant \omega_i^{ND}\\
\mathbb{P}\left[\omega_j>\omega_j^{ND}\text{ for all }j\neq i|\omega_i;e\right]\text{, if }\omega_i>\omega_i^{ND}.
\end{cases}$$
And consequently, for every $i\in N$, $d_i(\omega_i;e_N)=d_i(\omega_i;e_{N\setminus i})=0$ if $\omega_i\leqslant \omega_i^{ND}$ and $d_i(\omega_i;e_N)>d_i(\omega_i;e_{N\setminus i})$ if $\omega_i>\omega_i^{ND}$; where this inequality is due to $\omega_j>\omega_j^{ND}\text{ for all }j$ being an upper set of $\Omega_{N\setminus i}$ and the assumption that effort is team-improving.
Combining this inequality with the expression in (\ref{eq:11}), and dropping the dependence of $F_i(\cdot;e)$ on effort (since team-improving effort does not affect an agent's own marginal outcome distribution), we thus have
$$\mathbb{E}\left[\omega_i|ND;e_{N\setminus i}\right]=\frac{\int_{\Omega_i}\omega_i(1-d_i(\omega_i;e_{N\setminus i}))dF_i(\omega_i)}{\int_{\Omega_i}(1-d_i(\omega_i; e_{N\setminus i}))dF_i(\omega_i)}$$

$$=\frac{\int_{\omega_i\leqslant \omega_i^{ND}}\omega_i dF_i(\omega_i)+\int_{\omega_i> \omega_i^{ND}}\omega_i (1-d_i(\omega_i; e_{N\setminus i}))dF_i(\omega_i)}{\int_{\omega_i\leqslant \omega_i^{ND}}dF_i(\omega_i)+\int_{\omega_i> \omega_i^{ND}} (1-d_i(\omega_i; e_{N\setminus i}))dF_i(\omega_i)}$$

$$>\frac{\int_{\omega_i\leqslant \omega_i^{ND}}\omega_i dF_i(\omega_i)+\int_{\omega_i> \omega_i^{ND}}\omega_i (1-d_i(\omega_i; e_{N}))dF_i(\omega_i)}{\int_{\omega_i\leqslant \omega_i^{ND}}dF_i(\omega_i)+\int_{\omega_i> \omega_i^{ND}} (1-d_i(\omega_i; e_{N}))dF_i(\omega_i)}=\omega_i^{ND}=\mathbb{E}\left[\omega_i|ND;e_{N}\right].$$

And therefore, for each $i\in N$, the third term on the left-hand side of (\ref{eq:p1}) is strictly positive. And consequently, by Theorem \ref{th:1}, $fe(d')\subset fe(d)$ --- where $d$ is the equilibrium disclosure rule just described, given the consensual disclosure protocol and $d'$ is the full disclosure strategy. The statement then follows from the fact that full disclosure is the unique equilibrium disclosure strategy under the unilateral disclosure protocol --- and therefore $FE(D')\subset FE(D)$, where $D'$ is the consensual disclosure protocol and $D$ is the unilateral disclosure protocol.\qed
\subsubsection{Proof of Statement 3}
Statement 3 follows from the same argument in the proof of statement 2, and the observation that if the deliberation protocol is such that agent $i$ can unilaterally choose disclosure, then any equilibrium team-disclosure strategy involves full disclosure of $i$'s outcomes.\qed

\subsection{Proof of Proposition \ref{pr:4}}
\textbf{Step 1.} Fix $D'\neq D$, where $D$ is the unilateral disclosure protocol. 
As our first step in the proof, we observe (in Lemma \ref{lem:l2}) that if $\epsilon$ is sufficiently large, an equilibrium of the team-disclosure stage exists in which every team-member favors disclosure if and only if they do not draw their worst outcome. 
\begin{lemma}
\label{lem:l2}
There exists some $\epsilon'\in(0,1)$ such that if $\epsilon>\epsilon'$, there exists an equilibrium of the team-disclosure stage --- given full effort and deliberation procedure $D'$ --- where for every $i\in N$,
\begin{equation}
\label{eq:l6}x_i(\omega)=\begin{cases}0\text{, if }\omega_i=\ushort{\omega}_i\\
1\text{, otherwise.}\end{cases}
\end{equation}
\end{lemma}
\begin{proof}[Proof of Lemma]
Suppose an equilibrium of the team-disclosure stage exists in which individual disclosure strategies are as given in (\ref{eq:l6}); and suppose the implied equilibrium team-disclosure strategy is $d(\omega)=D(x(\omega))$. Then we have for each $i\in N$, and each $\epsilon\in(0,1)$,
\begin{align}
\omega_i^{ND,\epsilon}&=\mathbb{E}^\epsilon\left[\omega_i|ND;e_N\right]\nonumber\\[1em]
&=\mathbb{P}^\epsilon(\omega_i=\ushort{\omega_i}|ND;e_N)\ushort{\omega}_i+\mathbb{P}^\epsilon(\omega_i\neq\ushort{\omega_i}|ND;e_N)\mathbb{E}^\epsilon\left[\omega_i|ND,\omega_i\neq\ushort{\omega_i};e_N\right].\label{eq:l7}
\end{align}
But note that, given the individual disclosure strategies in (\ref{eq:l6}), no-disclosure happens only if at least one team-member $j\in N$ draws their worst possible outcome $\ushort{\omega}_j$. But note that, as $\epsilon\rightarrow 1$, it must be that for any $i,j\in N$, $\mathbb{P}(\omega_i=\ushort{\omega}_i|\omega_j=\ushort{\omega_j})\rightarrow1$. This, along with (\ref{eq:l7}) and the fact that $\mathbb{E}^\epsilon\left[\omega_i|ND,\omega_i\neq\ushort{\omega_i}\right]$ is bounded implies that for every $i\in N$,
\begin{equation}
\label{eq:l8}
\lim_{\epsilon\rightarrow1}\omega_i^{ND,\epsilon}=\ushort{\omega}_i.
\end{equation}
And consequently there is some $\epsilon'$ such that $\epsilon>\epsilon'$ implies that for every $i\in N$, $\omega_i^{ND,\epsilon}<\omega_i$ for all $\omega_i\in\Omega_i\setminus \{\ushort{\omega}_i\}$. And therefore the individual disclosure strategy in (\ref{eq:l6}) is individually rational and can be supported as an equilibrium of the team-disclosure stage.
\end{proof}
\vskip15pt
\noindent \textbf{Step 2.} Observe that, for $\epsilon>\epsilon'$, as given in Lemma \ref{lem:l2}, in the team-disclosure equilibrium described in the lemma, we have for some $i\in N$,
$$\mathbb{E}\left[\omega_i|ND;e_{N\setminus i}\right]>\ushort{\omega}_i.$$
And moreover, this value is independent of $\epsilon$. These statements are true because (1) $F(\cdot;e_{N\setminus i})$ has full support over $\Omega$ and is independent of $\epsilon$ for every $i\in N$; and (2) $D'$ is not the unilateral disclosure deliberation procedure, and therefore given the individual disclosure strategies in (\ref{eq:l6}) and $D'$, there exists some $i\in N$ and some $\omega\in\Omega$ with $\omega_i\neq\ushort{\omega}_i$ such that $d(\omega)=0$.
\vskip15pt
\noindent \textbf{Step 3.} Fix $\epsilon>\epsilon'$ as given in Lemma \ref{lem:l2} and consider the team-disclosure equilibrium described in the lemma. By equation (\ref{eq:l8}), and Step 2, we know that there is some $\bar{\epsilon}>\epsilon'$ such that, if $\epsilon>\bar{\epsilon}$,
$$\mathbb{E}\left[\omega_j|ND;e_{N\setminus j}\right]\geqslant \mathbb{E}\left[\omega_j|ND;e_{N}\right]$$
for all $j\in N$, and strictly so for some $i\in N$. 
\vskip15pt

\noindent \textbf{Step 4.} As a consequence of Step 3, and using Theorem \ref{th:3}, we know that if $\epsilon>\bar{\epsilon}$, $fe(d')\subset fe(d)$ --- where $d'$ is the full disclosure rule and $d$ is the equilibrium disclosure rule described in Lemma \ref{lem:l2}. Consequently, if $\epsilon>\bar{\epsilon}$, 
$$FE(D)\subset FE(D'),$$
and so the unilateral disclosure protocol is strictly dominated by $D'$.
\qed

\subsection{Proof of Proposition \ref{pr:5}}
This proposition follows almost immediately from Theorem \ref{th:3} and the definition of an effective team leader. The only missing step is noting that, if $i\in N$ is a team leader, then any equilibrium disclosure strategy is such that $d(\omega)<1$ only if $\omega_i=\ushort{\omega}_i$.\qed

\subsection{Proof of Proposition \ref{pr:eff_bin}}

Similar in spirit to Proposition \ref{pr:3}, the proof of Proposition \ref{pr:eff_bin} relies on Theorem \ref{th:3}. In equation (\ref{eq:p1}), the first two terms on the left-hand side are independent of the disclosure protocol. Hence, for a given cost vector, whether a disclosure protocol implements full effort or not relies on the last term. In particular, given that for full disclosure $\mathbb{P}(ND)=0$, comparing any K-majority rule to the full disclosure that follows from the unilateral disclosure protocol boils down to how effort affects $\mathbb{E}\left[\omega_i|ND\right]$. A disclosure protocol strictly dominates the unilateral disclosure protocol if and only if, given this particular disclosure protocol, 
$\mathbb{E}\left[\omega_i|ND; e_N\right]> \mathbb{E}\left[\omega_i|ND; e_{N\setminus i}\right]$. 

Given the binary structure and the symmetric deliberation protocol, for a given effort profile, we can rewrite $\mathbb{E}\left[\omega_i|ND \right]$ as $\frac{Pr(\omega_i=1 \cap ND)}{Pr(ND)}$. For a K-majority protocol, no-disclosure occurs if and only if at least $N-K+1$ team members are against disclosure, i.e. obtain a bad outcome. This can occur if either all receive the same common bad outcome or if at least $N-K+1$ team members receive independently bad draws of their individual binary outcome. Using this additional structure, we can rewrite the terms of he expressions as:\footnote{We adopt the convention that $\sum_{m=N}^{N-1}X(m)=0$ for any function $X.$ This is relevant if $K=1$, in which case, following a good outcome for player $i$ there is no possibility for no-disclosure.}
$$
\mathbb{P}(\omega_i=1 \cap ND) = (1-p)\mathbb{P}(h_i)\sum_{m=N-K+1}^{N-1} \binom{N-1}{m}(1-\mathbb{P}(h_j))^m \mathbb{P}(h_j)^{N-1-m}, \text{ and } $$
\begin{multline*}
 \mathbb{P}(ND) = p(1-\mathbb{P}(h_T))+(1-p)(1-\mathbb{P}(h_i))
\sum_{m=N-K}^{N-1} \binom{N-1}{m}(1-\mathbb{P}(h_j))^m \mathbb{P}(h_j)^{N-1-m} \\
+ (1-p)\mathbb{P}(h_i)\sum_{m=N-K+1}^{N-1} \binom{N-1}{m}(1-\mathbb{P}(h_j))^m \mathbb{P}(h_j)^{N-1-m}
\end{multline*}
\begin{multline*}
\Rightarrow \mathbb{P}(ND) = p(1-\mathbb{P}(h_T))+(1-p)
\sum_{m=N-K+1}^{N-1} \binom{N-1}{m}(1-\mathbb{P}(h_j))^m \mathbb{P}(h_j)^{N-1-m} \\
+ (1-p)(1-\mathbb{P}(h_i))\binom{N-1}{N-K}(1-\mathbb{P}(h_j))^{N-K}\mathbb{P}(h_j)^{K-1}
\end{multline*}
And so:
\begin{align}
\label{eq:e1}
\mathbb{E}\left[\omega_i|ND\right]=&\bigg[\frac{p(1-\mathbb{P}(h_T))}{(1-p)\mathbb{P}(h_i)\sum_{m=N-K+1}^{N-1} \binom{N-1}{m}(1-\mathbb{P}(h_j))^m \mathbb{P}(h_j)^{N-1-m}}\\
&\frac{1}{\mathbb{P}(h_i)}+\frac{(1-\mathbb{P}(h_i))\binom{N-1}{N-K}}{\mathbb{P}(h_i)\sum_{m=N-K+1}^{N-1}\binom{N-1}{m}\left(\frac{1-\mathbb{P}(h_j)}{\mathbb{P}(h_j)}\right)^{m-(N-K)}}\bigg]^{-1}\nonumber
\end{align}
In order to show the four statements, it suffices to sign the derivative of $\mathbb{E}\left[\omega_i|ND\right]$ with respect to the appropriate parameter. We begin with the first statement, so that we want to sign the derivative of $\mathbb{E}\left[\omega_i|ND\right]$ with respect to $\mathbb{P}(h_j)$. To do so, note that 
$$\sum_{m=N-K+1}^{N-1} \binom{N-1}{m}(1-\mathbb{P}(h_j))^m \mathbb{P}(h_j)^{N-1-m}$$
is decreasing in $\mathbb{P}(h_j)$, as it equals the probability of at least $N-K-1$ successes under a binomial with $N-1$ draws and success probability $1-\mathbb{P}(h_j)$. Moreover, 
$$\sum_{m=N-K+1}^{N-1}\binom{N-1}{m}\left(\frac{1-\mathbb{P}(h_j)}{\mathbb{P}(h_j)}\right)^{m-(N-K)}$$
is also decreasing in $\mathbb{P}(h_j)$, as $m>N-K$ for the whole range of summation. Consequently, we have that $\mathbb{E}\left[\omega_i|ND\right]$ is decreasing in $\mathbb{P}(h_j)$, and therefore under the parametrization in statement (i), we have $\mathbb{E}\left[\omega_i|ND,e_{N\setminus i}\right]>\mathbb{E}\left[\omega_i|ND,e_N\right]$. This implies that all symmetric deliberation protocols with $K>1$ strictly dominate the unilateral disclosure protocol (with $K=1$). 

Statements (ii)-(iv) follow from the same logic as statement (i), noting from equation (\ref{eq:e1}) that $\mathbb{E}\left[\omega_i|ND\right]$ decreases in $p$, and increases in $\mathbb{P}(h_i)$ and $\mathbb{P}(h_T)$. \qed






\section{Variations of the Benchmark Model}
\subsection{Continuous Outcome Distributions}
\label{app:B1}

In our benchmark team-disclosure model, we assume that the state-space $\Omega$ is finite and has a product structure. Suppose instead that the outcome distribution has full support and no mass points on $\Omega=[\ushort{\omega}_1,\bar{\omega}_1]\times...\times[\ushort{\omega}_N,\bar{\omega}_N]$. Then the following variation of Theorem \ref{th:1} holds.

\begin{theorem}
\label{th:1b}
The following statements are true about the equilibrium set: 
\begin{enumerate}
\item A full-disclosure equilibrium exists.
\item A partial-disclosure equilibrium exists only if disclosure cannot be chosen unilaterally by any team member. 
\item If a partial-disclosure equilibrium exists, then it is interior.
\end{enumerate}
\end{theorem}

A proof of Theorem \ref{th:1b} is available upon request. The statement differs from Theorem \ref{th:1} in two ways. First, the existence of a partial disclosure equilibrium is not guaranteed even if disclosure cannot be chosen unilaterally by any team-member. Rather, that condition is necessary but not sufficient for the existence of partial disclosure equilibria. A sufficient condition would be the existence of some $\omega\in \Omega$ such that $\Phi(\omega)>\omega$, where $\Phi$ is the mapping introduced in the proof of Theorem \ref{th:1}. This condition is met if disclosure cannot be chosen unilaterally by any team-member and outcomes are ``not too correlated'' across team-members. If the distribution of outcomes is discrete, then a product support is a sufficient notion of ``not too correlated,'' but that is not the case for continuous outcome distributions. 

The second difference between Theorems \ref{th:1} and \ref{th:1b} is that, when outcomes are distributed continuously, then all partial-disclosure equilibria are interior. To see why, note that if there is some agent for whom all outcomes except their worst possible outcome are disclosed, then it must be that all outcomes for all agents are disclosed (except perhaps their worst possible respective outcomes), because the locus $\omega_i=\ushort{\omega}_i$ has zero measure. 

\subsection{More General Individual Disclosure Strategies}
\label{app:B2}
Our equilibrium definition stated in Definition \ref{def:eq} requires each agent to use strategies that depend only on their own outcome realization. Theorem \ref{th:1c} below states a version of Theorem \ref{th:1} if we drop that equilibrium requirement. 

\begin{theorem}
\label{th:1c}
The following statements are true about the equilibrium set: 
\begin{enumerate}
\item A full-disclosure equilibrium exists.
\item A partial-disclosure equilibrium exists if and only if disclosure cannot be chosen unilaterally by all team members. 
\item An interior equilibrium exists if and only if disclosure cannot be chosen unilaterally.
\end{enumerate}
\end{theorem}

A proof of Theorem \ref{th:1c} is available upon request. The statement differs from Theorem \ref{th:1} in that, even if disclosure cannot be chosen unilaterally by any team members, there might be partial-disclosure equilibria which are not interior. To highlight the importance of agents' disclosure strategies conditioning only on their own outcomes for the third statement in Theorem \ref{th:1}, we provide a brief example.

Consider a team with three agents, A, B, and C. Each agent's outcome is binary, $\Omega=\{0,1\}^3$, and agents' outcomes are independent and equally likely.  The deliberation process is the majority rule: An outcome is disclosed if and only if at least two agents favor its disclosure. We now argue that there is a partial disclosure equilibrium that is not interior. Consider the following individual disclosure strategies: Agents A and B favor no disclosure if and only if the outcomes of both agents A and B are zero; Agent C always favors disclosure. The strategy of player A clearly depends on the outcome of player B, and vice-versa.

To see why this is an equilibrium, note that given the strategies above, the no-disclosure posteriors are given by the vector $(0,0,0.5)$ --- in case there is no disclosure, the outside observer knows that (i) agents A and B both had the worst possible realization; and (ii) there is no information about C's realization. Given these no-disclosure posteriors, it is easy to verify that there are no individual or coalitional deviations available, and therefore the individual disclosure strategies constitute an equilibrium. Note that, in this example, by allowing agents' individual disclosure strategies to depend on other team-members' outcomes, we can ``force'' team-members A and B to behave as if they were a single individual. By doing so, they remove all disclosure power from team-member C. But note instead that, if all agents' individual disclosure strategies were required to depend only on their own outcomes, then such an equilibrium would not exist. 

Finally note that if, instead of the majority deliberation protocol, the protocol were such that team-member $C$ has no power and team-members $A$ and $B$ both have the power to unilaterally veto disclosure, then the described outcome can be supported as an equilibrium in which all team-members use strategies that depend only on their own outcomes.


\end{document}